\def\epsdeltastuff{0}

\documentclass[letterpaper,11pt]{article}

\usepackage{preamble}

\newcommand{\llnorm}[1]{\left\lVert#1\right\rVert_2}
\newcommand{\abs}[1]{\left| #1 \right|}
\newcommand{\chisq}[1]{\chi^2\left( #1 \right)}

\newcommand{\id}{\mathbb{I}}

\newcommand{\PPDE}{\ensuremath{\textsc{PPDE}}}
\newcommand{\PGCE}{\ensuremath{\textsc{PGCENoBound}}}
\newcommand{\PGCEKappa}{\ensuremath{\textsc{PGCE}}}
\newcommand{\NaivePCE}{\ensuremath{\textsc{NaivePCE}}}
\newcommand{\PPreCond}{\ensuremath{\textsc{WeakPPC}}}
\newcommand{\PPC}{\ensuremath{\textsc{PPC}}}
\newcommand{\PPCRange}{\ensuremath{\textsc{PPCRange}}}
\newcommand{\PEstimateTrace}{\ensuremath{\textsc{PEstimateTrace}}}

\newcommand{\PPreCondRange}{\ensuremath{\textsc{WeakPPCNoBound}}}

\title{Privately Learning High-Dimensional Distributions}
\date{}
\author{
Gautam Kamath\thanks{Cheriton School of Computer Science, University of Waterloo. Work done when a graduate student at MIT, supported by NSF Award CCF-1617730, CCF-1650733, CCF-1741137, and ONR N00014-12-1-0999, and when a Microsoft Research Fellow, as part of the Simons-Berkeley Research Fellowship program. \href{mailto:g@csail.mit.edu}{\texttt{g@csail.mit.edu}}} \and 
Jerry Li\thanks{Microsoft Research AI. Work done when a graduate student at MIT, supported by NSF Award CCF-1453261 (CAREER), CCF-1565235, a Google Faculty Research Award, and an NSF Graduate Research Fellowship, and when a VMware Research Fellow, as part of the Simons-Berkeley Research Fellowship program. Part of this work was also performed while the author was an intern at Google. \href{mailto:jerrl@microsoft.com}{\texttt{jerrl@microsoft.com}}} \and 
Vikrant Singhal\thanks{Khoury College of Computer Sciences, Northeastern University. Supported by NSF award CCF-1750640.  \href{mailto:singhal.vi@husky.neu.edu}{\texttt{singhal.vi@husky.neu.edu}}} \and 
Jonathan Ullman\thanks{Khoury College of Computer Sciences, Northeastern University.  Supported by NSF awards CCF-1750640 and NSF awards CCF-1718088 and CNS-1816028, and a Google Faculty Research Award. \href{mailto:jullman@ccs.neu.edu}{\texttt{jullman@ccs.neu.edu}}}
}

\begin{document}

\maketitle

\pagenumbering{gobble}

\begin{abstract}
We present novel, computationally efficient, and differentially private algorithms for two fundamental high-dimensional learning problems: learning a multivariate Gaussian and learning a product distribution over the Boolean hypercube in total variation distance.  The sample complexity of our algorithms nearly matches the sample complexity of the optimal non-private learners for these tasks in a wide range of parameters, showing that privacy comes essentially \emph{for free} for these problems.  In particular, in contrast to previous approaches, our algorithm for learning Gaussians does not require strong \emph{a priori} bounds on the range of the parameters.  Our algorithms introduce a novel technical approach to reducing the sensitivity of the estimation procedure that we call \emph{recursive private preconditioning}.
\end{abstract}

\vfill
\newpage

\setcounter{tocdepth}{3}
\tableofcontents

\vfill
\newpage

\pagenumbering{arabic}
\section{Introduction}
\vspace{-6pt}
A central problem in machine learning and statistics is to learn (estimate) the parameters of an unknown distribution using samples.  However, in many applications, these samples consist of highly sensitive information belonging to individuals, and the output of the learning algorithm may inadvertently reveal this information.  While releasing only the estimated parameters of a distribution may seem harmless, when there are enough parameters---that is, when the data is \emph{high-dimensional}---these statistics can reveal a lot of individual-specific information (see e.g.~\cite{DinurN03,HomerSRDTMPSNC08,BunUV14,DworkSSUV15, ShokriSSS17}, and the survey~\cite{DworkSSU17}).  For example, the influential attack of Homer \etal~\cite{HomerSRDTMPSNC08} showed how to use very simple statistical information released in the course of genome wide association studies to detect the presence of individuals in those studies, which implies these individuals have a particular medical condition.  Thus it is crucial to design learning algorithms that ensure the privacy of the individuals in the dataset.

The most widely accepted solution to this problem is \emph{differential privacy}~\cite{DworkMNS06}, which provides a strong individual privacy guarantee by ensuring that no individual sample has a significant influence on the learned parameters.  A large body of literature now shows how to implement nearly every statistical algorithm privately, and differential privacy is now being deployed by Apple~\cite{AppleDP17}, Google~\cite{ErlingssonPK14}, and the US Census Bureau~\cite{DajaniLSKRMGDGKKLSSVA17}.

Differential privacy is typically achieved by adding noise to some non-private estimator, where the magnitude of the noise is calibrated to mask the effect of the single sample.  However, straightforward methods for adding this noise require strong \emph{a priori} bounds on the distribution to provide meaningful accuracy guarantees.  

\begin{quotation}
\noindent{\em Example:} Suppose we want to estimate the mean $\mu \in (-R,R)$ of a Gaussian random variable with known variance $1$, using the empirical mean of the data.  The empirical mean itself has variance $1/n$.  However, the na\"ive strategy for adding noise to the empirical mean would increase the variance by $O(R^2/n^2)$.  Thus, the variance of the na\"ive private algorithm dominates the variance in the data unless we have $\Omega(R^2)$ samples, forcing the user to have strong \emph{a priori} knowledge of the mean.
\end{quotation}

This problem is pervasive in applications of differential privacy, and considerable effort has been made to cope with this need for the parameters to lie in a small \emph{range}, including multiple systems that have been built to help elicit this information from the user~\cite{MohanTSSC12,GaboardiHKMNUV16}.  

When the distribution is \emph{low-dimensional}, there are many more effective algorithms that avoid the polynomial dependence on the size of the range.  For example, in the setting above of estimating a single univariate Gaussian, Karwa and Vadhan~\cite{KarwaV18} showed how to estimate the mean with \emph{essentially no} dependence on $R$.  More generally, there are also a plethora of general algorithmic techniques that can be used to address this general problem, such as \emph{smooth sensitivity}~\cite{NissimRS07}, \emph{subsample-and-aggregate}~\cite{NissimRS07,Smith11}, \emph{propose-test-release}~\cite{DworkL09}.  Unfortunately, none of these methods extend well to \emph{high-dimensional} problems.  When they exist, natural extensions either incur a costly dependence on the dimension in the sample complexity, or have running time exponential in the dimension.  

In this work we show how to privately learn two fundamental families of high-dimensional distributions with comparable costs to the corresponding optimal non-private learning algorithms:
\begin{itemize}
\item We give a computationally efficient algorithm for learning a multivariate Gaussian with unknown mean and covariance in total variation distance.  This algorithm requires only weak \emph{a priori} bounds on the mean and covariance, and in a wide range of parameters its sample complexity matches the optimal non-private algorithm up to lower-order terms.
\item We give a computationally efficient algorithm for learning a product distribution over the Boolean hypercube in total variation distance, which requires adding noise to each coordinate proportional to its variance, despite not knowing the variance \emph{a priori}.  Again, for many parameter regimes, the sample complexity of this algorithm is similar to that of the optimal non-private algorithm.
\end{itemize} 
Our results show that it is possible to obtain privacy nearly \emph{for free} when learning these important classes of high-dimensional distribution.  We obtain these results using a novel approach to reduce the sensitivity of the estimation procedure, which we call \emph{private recursive preconditioning}.

\subsection{Our Results}
\subsubsection{Privately Learning Gaussians} The most fundamental class of high-dimensional distributions is the \emph{multivariate Gaussian} in $\R^d$.  Our first result is an algorithm that takes samples from a distribution $\cN( \mu, \Sigma )$ with unknown mean $\mu \in \R^d$ and covariance $\Sigma \in \R^{d\times d}$ and estimates parameters $\wh\mu, \wh\Sigma$ such that $\cN(\wh{\mu}, \wh{\Sigma})$ is close to the true distribution in total variation distance (TV distance).  Without privacy, $n = \Theta(\frac{d^2}{\alpha^2})$ samples suffice to guarantee total variation distance at most $\alpha$ (this is folklore, but see e.g.~\cite{DiakonikolasKKLMS16}).

Despite the simplicity of this problem, it was only recently that Karwa and Vadhan~\cite{KarwaV18} gave an optimal algorithm for learning a univariate Gaussian.  Specifically, they showed that just $n = \tilde{O}(\frac{1}{\alpha^2} + \frac{1}{\alpha \eps} + \frac{\log(R\log \kappa)}{\eps})$ samples are sufficient to learn a univariate Gaussian $\normal(\mu, \sigma^2)$ with $|\mu| \leq R$ and $1\leq \sigma^2 \leq \kappa$, up to $\alpha$ in total variation distance subject to $\eps$-differential privacy.   In contrast to na\"ive approaches, their result has two important features: (1) The sample complexity has only mild dependence on the range parameters $R$ and $\kappa$, and (2) the sample complexity is only larger than that of the non-private estimator by a small multiplicative factor and an additive factor that is a lower order term for a wide range of parameters.  When the covariance is unknown, a na\"ive application of their algorithm would preserve neither of these features.

We show that it is possible to privately estimate a multivariate normal while preserving both of these features.  Our algorithms satisfy the strong notion of \emph{concentrated differential privacy (zCDP)}~\cite{DworkR16,BunS16}, which is formally defined in Section~\ref{sec:dp}.  To avoid confusion we remark that these definitions are on different scales so that $\frac{\eps^2}{2}$-zCDP is comparable to $\eps$-DP and $(\eps,\delta)$-DP.

\begin{thm}[Gaussian Estimation] \label{thm:maingaussian}
There is a polynomial time $\frac{\eps^2}{2}$-zCDP algorithm that takes
$$
n = \tilde{O}\left( \frac{d^2}{\alpha^2} + \frac{d^2}{\alpha \eps} + \frac{d^{3/2} \log^{1/2} \kappa + d^{1/2} \log^{1/2} R}{\eps} \right)
$$
samples from a Gaussian $\cN(\mu,\Sigma)$ with unknown mean $\mu \in \R^d$ such that $\|\mu\|_2 \leq R$ and unknown covariance $\Sigma \in \R^{d \times d}$ such that $\id \preceq \Sigma \preceq \kappa \id$, and outputs estimates $\wh\mu, \wh\Sigma$ such that, with high probability $\SD(\normal(\mu, \Sigma), \normal(\wh\mu,\wh\Sigma)) \leq \alpha$.  Here, $\tilde{O}(\cdot)$ hides polylogarithmic factors of $d, \frac{1}{\alpha}, \frac{1}{\eps}, \log \kappa,$ and $\log R$.  The same algorithm satisfies $(\eps \sqrt{\log(1/\delta)},\delta)$-differential privacy for every $\delta > 0$.
\end{thm}

Theorem~\ref{thm:maingaussian} will follow by combining Theorem~\ref{thm:pce-kappa} for covariance estimation with Theorem~\ref{thm:meanfinal} for mean estimation.  Observe that, since the sample complexity without privacy is $\Theta(\frac{d^2}{\alpha^2})$, Theorem~\ref{thm:maingaussian} shows that privacy comes almost for free unless $\frac{1}{\eps}, \kappa$, or $R$ are quite large.

The main difficulty that arises when trying to extend the results of~\cite{KarwaV18} to the multivariate case is that the covariance matrix of the Gaussian might be almost completely \emph{unknown}.  The main technically novel part of our algorithm is a method for learning a matrix $A$ approximating the inverse of the covariance matrix so that $\id \preceq A \Sigma A \preceq 1000 \id$.  This matrix can be used to transform the Gaussian to be nearly spherical, making it possible to apply the methods of~\cite{KarwaV18}.

\begin{thm}[Private Preconditioning] \label{thm:maincovariance}
There is an $\frac{\eps^2}{2}$-zCDP algorithm that takes
$$
n = \tilde{O}\left(\frac{d^{3/2} \log^{1/2} \kappa}{\eps}\right) 
$$
samples from an unknown Gaussian $\cN(0,\Sigma)$ over $\R^d$ with $\Sigma \in \R^{d \times d}$ such that $\id \preceq \Sigma \preceq \kappa \id$, and outputs a symmetric matrix $A$ such that $\id \preceq A \Sigma A \preceq 1000 \id$.  Here, $\tilde{O}(\cdot)$ hides polylogarithmic factors of $d, \frac{1}{\eps}$, and $\log \kappa$.
\end{thm}

We describe and analyze our algorithm for private covariance estimation in Section~\ref{sec:gaussian-cov}, and in Section~\ref{sec:gaussian-mean} we combine it with the algorithms of~\cite{KarwaV18} to obtain Theorem~\ref{thm:maingaussian}.

\ifnum\epsdeltastuff=1
\medskip
\emph{Learning Unbounded Gaussians.}
Theorem~\ref{thm:maingaussian} and Theorem~\ref{thm:maincovariance} have only a polylogarithmic dependence on $R,\kappa$, which is necessary for any algorithm satisfying zCDP.  For univariate Gaussians,~\cite{KarwaV18} also showed how to entirely remove the dependence on these parameters subject to $(\eps,\delta)$-DP.  Specifically, they show that for any Gaussian $\normal(\mu, \sigma^2)$, it is possible to learn up to distance $\alpha$ with $(\eps,\delta)$-DP using just $n = \tilde{O}(\frac{1}{\alpha^2} + \frac{1}{\alpha \eps} + \frac{\log(1/\delta)}{\eps})$ samples.  As before, when the covariance is known, we can relatively easily extend this result to the multivariate case.  Our next result shows how to learn unbounded multivariate Gaussians with unknown covariance.

\begin{thm}\label{thm:maingaussianunbounded}
There is a polynomial time $(\eps,\delta)$-DP algorithm that takes
$$
n = \tilde{O}\left( \frac{d^2}{\alpha^2} + \frac{d^2 \sqrt{\log(1/\delta)}}{\alpha \eps}\right)
$$
samples from a Gaussian $\cN(\mu,\Sigma)$ with unknown mean $\mu \in \R^d$ and covariance $\Sigma \in \R^{d \times d}$ with $\Sigma \succeq \id$, and outputs $\wh\mu, \wh\Sigma$ such that, with high probability $\SD(\normal(\mu, \Sigma), \normal(\wh\mu,\wh\Sigma)) \leq \alpha$.  Here, $\tilde{O}(\cdot)$ hides polylog factors of $d, \frac{1}{\alpha}, \frac{1}{\eps}$, and $\log(\frac{1}{\delta})$.
\end{thm}

Theorem~\ref{thm:maingaussianunbounded} will follow by combining Theorem~\ref{thm:pce-nokappa} for covariance estimation with Theorem~\ref{thm:meanfinal} for mean estimation.  The main technical ingredient is an analogue of the recursive private preconditioner of Theorem~\ref{thm:maincovariance} that allows us to take a Gaussian $\normal(0,\Sigma)$ with unbounded condition number and find a preconditioning matrix $A$ such that $\id \preceq A \Sigma A \preceq O(d^4) \id$.  Once we have such a preconditioner, we can apply Theorem~\ref{thm:maingaussian}.  We remark that the assumption $\Sigma \succeq \id$ in Theorem~\ref{thm:maingaussianunbounded} is arbitrary, and can be replaced with the assumption $\Sigma \succeq \gamma \id$ for any $\gamma > 0$.  We do however require some lower bound on the smallest eigenvalue of $\Sigma$.
\fi

\subsubsection{Privately Learning Product Distributions}  The simplest family of high-dimensional discrete distributions are product distributions over $\zo^{d}$.  Without privacy, $\Theta(\frac{d}{\alpha^2})$ are necessary and sufficient to learn up to $\alpha$ in total variation distance.  The standard approach to achieving DP by perturbing each coordinate independently requires $\tilde{\Theta}(\frac{d}{\alpha^2} + \frac{d^{3/2}}{\alpha \eps})$ samples.  We give an improved algorithm for this problem that avoids this blowup in sample complexity.  While our algorithm for learning product distributions is quite different to our algorithm for estimating Gaussian covariance, it uses a similar recursive preconditioning technique, highlighting the versatility of this approach.

\begin{thm} \label{thm:mainproduct}
There is a polynomial time $\frac{\eps^2}{2}$-zCDP algorithm that takes
$$
n = \tilde{O}\left(\frac{d}{\alpha^2} + \frac{d}{\alpha \eps}\right)
$$
samples from an unknown product distribution $\cP$ over $\zo^{d}$ and outputs a product distribution $\cQ$ such that, with high probability, $\SD(\cP,\cQ) \leq \alpha$.  Here, $\tilde{O}(\cdot)$ hides polylogarithmic factors of $d, \frac{1}{\alpha},$ and $\frac{1}{\eps}$.
The same algorithm satisfies $(\eps \sqrt{\log(1/\delta)},\delta)$-differential privacy for every $\delta > 0$.
\end{thm}

We describe and analyze our algorithm in Section~\ref{sec:product}.  

\subsubsection{Lower Bounds}
We prove lower bounds for the problems we consider in this paper, demonstrating that for many problems, our sample complexity is optimal up to polylogarithmic factors.
One example statement is the following lower bound for private mean estimation of a product distribution, for the more permissive notion of $(\eps,\delta)$-differential privacy (compared to our upper bounds, which are in terms of zCDP):

\begin{thm} \label{thm:mainproductlowerbound}
Any $(\eps, \frac{1}{64n})$-differentially private algorithm that takes samples from an arbitrary unknown product distribution $\cP$ over $\zo^{d}$ and outputs a product distribution $\cQ$ such that $\SD(\cP,\cQ) \leq \alpha$ with probability $\geq 9/10$ requires $n = \Omega(\frac{d}{\alpha^2} + \frac{d}{\alpha \eps \log d})$ samples.
\end{thm}
We also prove a qualitatively similar lower bound for privately estimating the mean of a Gaussian distribution.

In addition, we prove lower bounds for privately estimating a Gaussian with unknown covariance.
These are qualitatively weaker, as they are only for $\eps$-differential privacy, and we consider it an interesting open question to prove lower bounds for covariance estimation under concentrated or approximate differential privacy.
\begin{thm}
Any $\eps$-differentially private algorithm that takes samples from an arbitrary unknown Gaussian distribution $\cP$ and outputs a Gaussian distribution $\cQ$ such that $\SD(\cP, \cQ) \leq \alpha$ with probability $\geq 9/10$ requires $n = \Omega(\frac{d^2}{\alpha^2} + \frac{d^2}{\alpha \eps})$ samples.
\end{thm}

The last question to address is the dependence on the parameters $R$ and $\kappa$.
It is well-known that, under zCDP, the existence of a $0.1$-packing (in total variation distance) $\mathcal{P}$ results in a lower bound of $n = \Omega(\frac{1}{\eps} \log^{1/2}|\mathcal{P}|)$~\cite{HardtT10, BeimelBKN14,BunS16}.
Since, for a set of identity covariance Gaussians, there exists such a packing of size $R^{\Omega(d)}$, this implies that the dependence of Theorem~\ref{thm:maingaussian} on $R$ is optimal up to logarithmic factors.
As for the dependence on $\kappa$, it can be shown that for $d = 2$, there exists a packing of zero-mean Gaussians of size $\poly(\kappa)$~\cite{Kamath19}, in contrast to $\poly \log \kappa$ in one dimension.
That is, though our algorithm's dependence on $\kappa$ is exponentially greater than that of Karwa and Vadhan~\cite{KarwaV18}, this is necessary for any $d \geq 2$.

All our lower bounds are presented in Section~\ref{sec:lb}.

\subsubsection{Comparison to Lower Bounds for High-Dimensional DP}
Readers familiar with differential privacy may wonder why our results do not contradict known lower bounds for high-dimensional estimation in differential privacy~\cite{BunUV14,SteinkeU15,DworkSSUV15}.  The two key differences are (1) lower bounds showing that privacy is costly are for the relatively weak $\ell_\infty$ estimation guarantee whereas we want a rather stringent estimation guarantee, and (2) we exploit the structure of Gaussians and product distributions to obtain guarantees that are not possible for arbitrary distributions.

To understand the first issue, most lower bounds in differential privacy apply to estimating the mean of the distribution up to $\alpha$ in $\ell_\infty$ distance.  This guarantee can be achieved with $\Theta(\frac{\log d}{\alpha^2})$ samples non-privately but requires $\Theta(\frac{\log d}{\alpha^2} + \frac{\sqrt{d}}{\alpha})$ samples with differential privacy.  Thus, for $\ell_\infty$ estimation, privacy is costly in high dimensions.  However, if we consider the more stringent $\ell_2$ metric, then the cost of privacy goes away, and $\Theta(\frac{d}{\alpha^2})$ samples are sufficient with or without privacy.\footnote{One way to see this is that estimation up to $\alpha/\sqrt{d}$ in $\ell_\infty$ implies estimation up to $\alpha$ in $\ell_2$, so the non-private term and the private term in the $\ell_\infty$ bounds have roughly the same dependence on the dimension in this case.}  Thus, for this stronger guarantee, privacy is not costly in high-dimensions.  This phenomenon is fairly general, and is not specific to Gaussians or product distributions.

To understand the second issue, in order to learn Gaussians or product distributions in total variation distance, we need to learn in metrics that are related to $\ell_2$, but take into account the variance of the distribution.  In the case of Gaussians we learn in the Mahalanobis distance $\| \cdot \|_{\Sigma}$ and for product distributions our guarantees are closely related to $\chi^2$-divergence.  For example, Theorem~\ref{thm:maingaussian} can actually be rephrased as saying that, for our algorithm,
$$
n = \tilde{O}\left(\frac{d^{3/2} \log^{1/2} \kappa}{\eps}\right) \Longrightarrow \| \Sigma - \wh \Sigma \|_{\Sigma} = O\left( \sqrt{\frac{d^2}{n}} + \frac{d^2}{\eps n} \right).
$$
This sort of guarantee where the error in the $\Sigma$-norm does not depend on the range parameter $\kappa$ cannot be achieved for arbitrary distributions, and thus for this part of the guarantee we crucially use the fact that the data is i.i.d.\ from a Gaussian.  A similar phenomenon arises for product distributions, where, as we show, learning the mean of a product distribution in the right metric can be done with $\tilde{O}(d)$ samples for product distributions but would require $\Omega(d^{3/2})$ samples for arbitrary distributions.

\subsection{Techniques}

\subsubsection{Privately Learning Gaussians}  

We now give an overview of the main ideas that go into estimating multivariate Gaussians (Theorem~\ref{thm:maingaussian}).  We make two simplifications to ease the presentation.  First, we assume the distribution has mean zero, and focus only on covariance estimation.  Second, we elide the accuracy and privacy parameters, $\alpha$ and $\eps$, since they do not play a central role in the discussion.

Suppose we are given samples $X_1,\dots,X_n \sim \normal(0,\Sigma)$ and want to output $\wh\Sigma$ such that $\normal(0,\wh\Sigma)$ is close to the true distribution in TV distance.  More precisely, we want to guarantee
\begin{equation} \label{eq:sigmaconvergence}
\| \wh\Sigma - \Sigma \|_{\Sigma} = \| \Sigma^{-1/2} \wh\Sigma \Sigma^{-1/2} - \id \|_{F} \lesssim \frac{1}{100},
\end{equation}
which implies closeness in TV distance.  Then the standard solution is to use the empirical covariance $\wt\Sigma = \frac{1}{n} \sum_i X_i X_i^T$, guaranteeing $\| \wt\Sigma  - \Sigma \|_{\Sigma} \lesssim \sqrt{d^2/n}$, satisfying \eqref{eq:sigmaconvergence} when $n \gtrsim d^2$.

\medskip\textbf{First Attempt: The Gaussian Mechanism.}
The standard way to privately estimate a function $f$, is to use the \emph{Gaussian mechanism}.  In this case we are interested in the matrix-valued empirical covariance $\wt\Sigma(X_1,\dots,X_n) = \frac{1}{n} \sum_{i} X_i X_i^T$, so the Gaussian mechanism would output
$$
\wh{\Sigma} = \wt\Sigma(X_1,\dots,X_n) + Z
~~\textrm{where $Z \sim \normal(0, O(\Delta_{\wt\Sigma}^2))^{d \times d}$}$$
is a random Gaussian \emph{noise matrix} and $\Delta_{\wt\Sigma} = \max_{X \sim X'} \| \wt\Sigma(X) - \wt\Sigma(X') \|_{F}$ where $X \sim X'$ denotes that $X,X'$ differ on at most one sample and $\Delta_{\wt\Sigma}$ is called the \emph{global sensitivity}.  Note that it is only a coincidence that the true distribution and the noise distribution are both Gaussian.

Unfortunately, the empirical covariance has \emph{infinite sensitivity}, since $X_i$ is an arbitrary vector, and thus changing a single sample can change the empirical covariance arbitrarily.  The simplest way to address this problem is to assume that we have some \emph{prior information} about $\Sigma$, namely that $\id \preceq \Sigma \preceq \kappa \id$.  In this case we can \emph{clamp} every sample so that $\| X_i \|_2^2 \leq  \tilde{O}(\kappa d)$.  Once we do this, the sensitivity of the clamped empirical covariance is $\tilde{O}(\kappa d / n)$.  However, if the data really came from a Gaussian, then the clamping will not have any effect, as long as there are no significant outliers.  Thus, when we apply the Gaussian mechanism we obtain the guarantee
\begin{equation} \label{eq:analyzegauss}
\left\| \wh\Sigma - \Sigma \right\|_{\Sigma} \leq \left\| \wh{\Sigma} - \Sigma \right\|_{F} = \tilde{O}\left(\sqrt{\frac{d^2}{n}} +  \frac{\kappa d^2}{n}\right),
\end{equation}
where the first inequality follows from the assumption that $\Sigma \succeq \id$ and the second uses a standard bound on the Frobenius norm of a random Gaussian matrix.  Unfortunately,~\eqref{eq:analyzegauss} has a linear dependence on $\kappa$, meaning the number of samples has to be at least $\Omega(\kappa d^2)$ to achieve~\eqref{eq:sigmaconvergence}, and this analysis is tight when, say, $\Sigma = \id$.

Before moving on, we highlight the fact that, due to truncation, this algorithm ensures privacy \emph{for any dataset}, even one not drawn from a Gaussian.  This is a critical feature in private estimation that our algorithms will preserve.

\medskip\textbf{Recursive Private Preconditioning.}
Looking at~\eqref{eq:analyzegauss}, we can see that the Gaussian mechanism actually has excellent accuracy when $\kappa = O(1)$, specifically the term corresponding to privacy vanishes faster than the term corresponding to sampling error.  Thus, our approach to improve over the Gaussian mechanism is to private find a symmetric \emph{preconditioner} $A$ so that $\id \preceq A \Sigma A \preceq O(1) \id$.  Given such a matrix, we can apply the Gaussian mechanism to the data $AX_1,\dots,AX_n \sim \normal(0, A \Sigma A)$, and learn this transformed distribution using just $\tilde{O}(d^2)$ samples.  Of course, to be useful, we need to find the private preconditioner using $\tilde{O}(d^2)$ samples.

In order to obtain such a matrix $A$, we essentially need a good multiplicative estimate of the covariance matrix $\Sigma$ along every direction.  However, the Gaussian mechanism makes this difficult because it adds an i.i.d.\ Gaussian matrix $Z$, which ignores the shape of $\Sigma$.  Specifically, if we use the Gaussian mechanism and add the noise matrix $Z$, then unless we draw $\Omega(\kappa d^{3/2})$ samples, the directions of low variance will be completely overwhelmed by noise.

Our main observation is that when we use the Gaussian mechanism, even with just $n = \tilde{O}(d^{3/2})$ samples, we can still obtain a good enough estimate of $\Sigma$ to make some progress.  Specifically, we can find a matrix $A$ such that $\id \preceq A \Sigma A \preceq \frac{7}{10} \kappa \id$.  Given such a procedure, we can iterate $O(\log \kappa)$ times until the condition number is a constant.

To see how we make progress, we argue that if we draw just $\tilde{O}(d^{3/2})$ samples, and compute the noisy empirical covariance $\wh\Sigma = \wt\Sigma + Z$, then the directions of $\wt\Sigma$ with \emph{large variance} are approximately preserved, even though the directions with \emph{small variance} will be overwhelmed by noise.  We can leverage this fact in the following way.  Suppose we see a direction of $\wh\Sigma$ with variance at least $\kappa/2$.  Then this direction cannot only appear to have large variance due to noise, it must also be large in $\Sigma$, meaning we have a good multiplicative approximation.  On the other hand, suppose we see a direction of $\wh\Sigma$ with variance at most $\kappa/2$.  Then this direction may have almost no variance in $\Sigma$, but it cannot possibly have variance much larger than $\kappa/2$.  Thus, we have discovered that this direction has less variance than our bound $\kappa$, meaning we can clamp the samples more aggressively in that direction to reduce the sensitivity of the estimator!  More precisely, we compute the eigenvectors and eigenvalues of $\wh\Sigma$ and let $A$ be the matrix that partially projects out the eigenvectors with large eigenvalues.  We can show that this matrix reduces the maximum variance by more than it reduces the minimum variance, thus (after some rescaling) $\id \preceq A \Sigma A \preceq \frac{7}{10} \kappa \id$, as desired.  

\medskip\textbf{Connection to Average Sensitivity.}
At a high level, the problem with the Gaussian mechanism is that it adds noise proportional to the \emph{global sensitivity}, which is the maximum that changing one sample can change the estimate in the worst case.  Intuitively, we'd like to add noise proportional to the \emph{average sensitivity}, which is the amount that replacing one sample from the true distribution with an independent random sample from the true distribution changes the estimate.  Simply adding noise proportional to average sensitivity is not private, however there are several techniques (e.g.~\cite{NissimRS07, DworkL09}) that make it possible to add noise roughly proportional to average sensitivity for \emph{univariate} statistics.  Unfortunately, these techniques either do not apply, or are computationally inefficient for \emph{multivariate} statistics.  Our recursive private preconditioning technique can be viewed as a new technique for adding noise proportional to average sensitivity that is computationally efficient in high dimensions.  Currently the technique is specific to Gaussian (or subgaussian) covariance estimation, but it would be interesting to understand how generally this technique can be applied.

\subsubsection{Privately Learning Product Distributions}   Although learning Boolean product distributions is quite different from learning Gaussians, our algorithm uses a similar approach of recursively reducing the sensitivity of the natural estimator.  Suppose we have a product distribution $\cP$ over $\zo^d$ with mean $p = \ex{}{\cP}$ and want to output a product distribution $\wh\cP$ with mean $\wh p$ that is close in TV distance.  Bounding the TV distance between $\cP,\wh\cP$ requires some care, so to simplify this high-level discussion, we assume that $p \succeq 1/d$, in which case
\begin{equation} \label{eq:productconvergence}
\SD(\cP, \wh\cP) \leq O\left( \sqrt{\sum_j \frac{(p_j - \wh p_j)^2}{p_j}} \right)
\end{equation}
is a reasonably tight bound.  Without privacy, it would suffice to draw samples $X_1,\dots,X_n \sim \cP$ and compute the empirical mean $\wt p = \frac{1}{n} \sum_{i} X_i$ and output the corresponding product distribution, ensuring $\SD(\cP,\wh\cP) \lesssim \sqrt{d/n}$.

\medskip\textbf{First Attempt: The Gaussian Mechanism.}  As with Gaussians, the natural approach is to use the Gaussian mechanism, and compute $\wh p = \wt p + Z$ where $Z \sim \normal(0, (\sqrt{d}/n)^2 )^d$ is a vector with i.i.d.\ Gaussian entries and $\sqrt{d}/n$ is the global sensitivity of the empirical mean.  The problem with this mechanism is that if the mean of $\cP$ is roughly $1/d$ in every coordinate, then we cannot upper bound~\eqref{eq:productconvergence} unless $\|Z\|_2 \lesssim 1/\sqrt{d}$, which requires $n = \Omega(d^{3/2})$.


\medskip\textbf{Recursive Private Preconditioning.}
The key to improving the sample complexity is to recognize that the sensitivity analysis and the accuracy analysis cannot both be tight at the same time.  That is, if $p = (\frac12,\dots,\frac12)$ then it suffices to have $\| Z \|_2 \lesssim 1$, which requires only $n = O(d)$ samples.  On the other hand, if $p = (\frac1d,\dots,\frac1d)$, then we really need $\| Z \|_2 \lesssim 1/\sqrt{d}$.  However, in this case each sample from $\cP$ will satisfy $\| X_i \|_2 = O(\log d)$, so we reduce the sensitivity to just $O(\log d / n)$ by clamping  the samples to have this norm, in which case we can obtain $\|Z\|_2 = 1/\sqrt{d}$ using just $n = \tilde{O}(d)$ samples.

The challenge is that $p$ may have some coordinates that are roughly balanced and some that are very biased.  If we could partition the coordinates into groups based on their bias, then we could apply an argument like the above on each group separately, but the challenge is to do this partitioning privately.  

Similar to what we did for Gaussians, we can achieve this partitioning by starting with the Gaussian mechanism.  Suppose we use the Gaussian mechanism to obtain $\wh p = \wt p + Z$.  Then if we draw $n = \tilde{O}(d)$ samples, we will have $\| Z \|_\infty = \tilde{O}(1/\sqrt{d})$.  Now, consider two cases:  for coordinates $j$ such that $\wh p_j \geq 1/4$, then we know that $\wt p_j$ is at least, say, $1/8$, so we have
$
(\wh p_j - \wt p_j)^2 / \wt p_j = O(1/d)
$
which is a good enough estimate for that coordinate.  Thus, we can lock in our estimate of these coordinates and move on.  Now, the coordinates $j$ we have left satisfy $\wh p_j \leq 1/4$, and thus we know $\wt p_j$ is at most, say, $3/8$.  Thus, if we restrict the distribution to just these coordinates, then we can clamp the norm of the samples and estimate again using less noise.  Every time we iterate this process, we can reduce the upper bound on the bias by a constant factor until we get down to the case where all coordinates have bias at most $O(1/d)$, which we can handle separately.  Iterating this approach $O(\log d)$ times requires us to add up estimation error and privacy loss across the different rounds, but this only incurs additional polylogarithmic factors.

\subsection{Additional Related Work}


\medskip\emph{Differentially Private Learning and Statistics.}  The most directly comparable papers to ours are recent results on learning \emph{low-dimensional} statistics.  In addition to the aforementioned work of Karwa and Vadhan~\cite{KarwaV18}, Bun, Nissim, Stemmer, and Vadhan~\cite{BunNSV15} showed how to private learn an arbitrary distribution in Kolmogorov distance, which is weaker than TV distance, with almost no increase in sample complexity.  Diakonikolas, Hardt, and Schmidt~\cite{DiakonikolasHS15} extended that work to give a practical algorithm for learning structured one-dimensional distributions in TV distance.

An elegant work of Smith~\cite{Smith11} showed how to estimate arbitrary \emph{asymptotically normal} statistics with only a small increase in sample complexity compared to non-private estimation.  Technically, this work doesn't apply to covariance estimation or estimating sparse product distributions, for which the asymptotic distribution is not normal.  More fundamentally, this algorithm learns a high-dimensional distribution one coordinate at a time, which is quite costly for the distributions we consider here.

Subsequent to our work, Cai, Wang, and Zhang~\cite{CaiWZ19} studied mean and covariance estimation of subgaussian distributions (as well as sparse mean estimation, which we don't consider in this work) subject to differential privacy, but in a setting with strong \emph{a priori} bounds on the parameters.  In particular, they prove a lower bound for mean estimation of subgaussian distributions that is incomparable to our Theorems~\ref{thm:product-lb} and~\ref{thm:guassian-lb}.  It is quantitatively larger by a factor of $\log^{1/2}(1/\delta)$, but it only holds for the more general class of subgaussian non-product distributions.  In particular they use a reduction from~\cite{SteinkeU17a} to boost one of Theorem~\ref{thm:product-lb} or~\ref{thm:guassian-lb} in a way that fails to preserve the property of being Gaussian or being a product distribution.


\medskip
\emph{Covariance Estimation.}  For covariance estimation, the works closest to ours are that of Dwork, Talwar, Thakurta, and Zhang~\cite{DworkTTZ14} and Sheffet~\cite{Sheffet17}.  Their algorithms require that the norm of the data be bounded, and the sample complexity depends polynomially on this bound.  In contrast, our algorithms have either mild or no dependence on the norm of the data.

\medskip
\emph{Robust Statistical Estimation on High-Dimensional Data.} 
Recently, there has been significant interest in the computer science community in robustly estimating distributions~\cite{DiakonikolasKKLMS16, LaiRV16, CharikarSV17, DiakonikolasKKLMS17, DiakonikolasKKLMS18, SteinhardtCV18}, where the goal is to estimate some distribution from samples even when a constant fraction of the samples may be corrupted by an adversary.  As observed by Dwork and Lei~\cite{DworkL09}, differentially private estimation and robust estimation both seek to minimize the influence of outliers, and thus there is a natural conceptual connection between these two problems.  Technically, the two problems are incomparable. Differential privacy seeks to limit the influence of outliers in a very strong sense, and without making any assumptions on the data, but only when up to $O(1/\eps)$ samples are corrupted.  In contrast, robust estimation limits the influence of outliers in a weaker sense, and only when the remaining samples are chosen from a nice distribution, but tolerates up to $\Omega(n)$ corruptions.

\medskip\emph{Differentially Private Testing.}  There have also been a number of works on differentially private \emph{hypothesis testing}.  For example,~\cite{WangLK15, GaboardiLRV16,KiferR17, CaiDK17, KakizakiSF17} gave private algorithms for goodness-of-fit testing, closeness, and independence testing.  Recently,~\cite{AcharyaSZ18} and~\cite{AliakbarpourDR18} have given essentially optimal algorithms for goodness-of-fit and closeness testing of arbitrary distributions.  \cite{AcharyaKSZ18}~designed nearly optimal algorithms for estimating properties like support size and entropy. \cite{GaboardiR18,Sheffet18} study hypothesis testing in the local differential privacy setting. All these works consider testing of \emph{arbitrary} distributions, and so they necessarily have sample complexity growing exponentially in the dimension.


\medskip
\emph{Privacy Attacks and Lower Bounds.} A complementary line of work has established limits on the accuracy of private algorithms for high-dimensional learning.  For example, Dwork \etal~\cite{DworkSSUV15} (building on~\cite{BunUV14, HardtU14, SteinkeU15, SteinkeU17a}) designed a \emph{robust tracing attack} that can infer sensitive information about individuals in a dataset using highly noisy statistical information about the dataset.  These attacks apply to nice distributions like product distributions and Gaussians, but require that the dataset be too small to learn the underlying distribution in total variation distance, and thus do not contradict our results.  These attacks apply to a number of learning problems, such as PCA~\cite{DworkTTZ14}, ERM~\cite{BassilyST14}, and variable selection~\cite{BafnaU17,SteinkeU17b}.  Similar attacks lead to computational hardness results for differentially private algorithms for high-dimensional data~\cite{DworkNRRV09,UllmanV11,Ullman16,KowalczykMUZ16,KowalczykMUW18}, albeit for learning problems that encode certain cryptographic functionalities.



\section{Preliminaries} \label{sec:dp}
A \emph{dataset} $X = (X_1,\dots,X_n) \in \cX^n$ is a collection of elements from some \emph{universe}.  We say that two datasets $X,X' \in \cX^n$ are \emph{neighboring} if they differ on at most a single entry, and denote this by $X \sim X'$.  Informally, differential privacy requires that for every pair of datasets $X,X' \in \cX^n$ that differ on at most a single entry, the distributions $M(X)$ and $M(X')$ are close.  In our work we consider a few different variants of differential privacy.  The first is the standard variant of differential privacy.
\begin{defn}[Differential Privacy (DP)~\cite{DworkMNS06}] A randomized algorithm $M: \cX^n \rightarrow \cY$ satisfies \emph{$(\eps,\delta)$-differential privacy ($(\eps,\delta)$-DP)} if for every pair of neighboring datasets $X, X' \in \cX^n$,
$$
\forall Y \subseteq \cY~~~\pr{}{M(X) \in Y} \leq e^{\eps} \pr{}{M(X') \in Y} + \delta.
$$ 
\end{defn}

\noindent The second variant is so-called \emph{concentrated differential privacy}~\cite{DworkR16}, specifically the refinement \emph{zero-mean concentrated differential privacy}~\cite{BunS16}.
\begin{defn}[Concentrated Differential Privacy (zCDP)~\cite{BunS16}]
    A randomized algorithm $M: \cX^n \rightarrow \cY$
    satisfies \emph{$\rho$-zCDP} if for
    every pair of neighboring datasets $X, X' \in \cX^n$,
    $$\forall \alpha \in (1,\infty)~~~D_\alpha\left(M(X)||M(X')\right) \leq \rho\alpha,$$
    where $D_\alpha\left(M(X)||M(X')\right)$ is the
    $\alpha$-R\'enyi divergence between $M(X)$ and $M(X')$.\footnote{Given two probability distributions $P,Q$ over $\Omega$, $D_{\alpha}(P\|Q) = \frac{1}{\alpha - 1}\log\left( \sum_{x} P(x)^{\alpha} Q(x)^{1-\alpha}\right)$}
\end{defn}

\noindent Both of these definitions are closed under post-processing
\begin{lem}[Post-Processing~\cite{DworkMNS06,BunS16}] \label{lem:postprocessing}
If $M : \cX^n \to \cY$ is $(\eps,\delta)$-DP and $P : \cY \to \cZ$ is any randomized function, then the algorithm $P \circ M$ is $(\eps,\delta)$-DP.  Similarly if $M$ is $\rho$-zCDP then the algorithm $P \circ M$ is $\rho$-zCDP.
\end{lem}

\noindent Qualitatively, DP with $\delta = 0$ is stronger than zCDP, which is stronger than DP with $\delta > 0$.  These relationships are quantified in the following lemma.
\begin{lem}[Relationships Between Variants of DP~\cite{BunS16}] \label{lem:dpdefns}
 For every $\eps \geq 0$,
\begin{enumerate}
\item If $M$ satisfies $(\eps,0)$-DP, then $M$ is $\frac{\eps^2}{2}$-zCDP.
\item If $M$ satisfies $\frac{\eps^2}{2}$-zCDP, then $M$ satisfies $(\frac{\eps^2}{2} + \eps \sqrt{2 \log(\frac{1}{\delta})},\delta)$-DP for every $\delta > 0$.
\end{enumerate}
\end{lem}

\noindent Note that the parameters for DP and zCDP are on different scales, with $(\eps,\delta)$-DP roughly commensurate with $\frac{\eps^2}{2}$-zCDP.

\medskip
\noindent\textbf{Composition.} A crucial property of all the variants of differential privacy is that they can be composed adaptively.  By adaptive composition, we mean a sequence of algorithms $M_1(X),\dots,M_T(X)$ where the algorithm $M_t(X)$ may also depend on the outcomes of the algorithms $M_1(X),\dots,M_{t-1}(X)$.
\begin{lem}[Composition of DP~\cite{DworkMNS06, DworkRV10, BunS16}] \label{lem:dpcomp}
If $M$ is an adaptive composition of differentially private algorithms $M_1,\dots,M_T$, then the following all hold:
\begin{enumerate}
\item If $M_1,\dots,M_T$ are $(\eps_1,\delta_1),\dots,(\eps_T,\delta_T)$-DP then $M$ is $(\eps,\delta)$-DP for $$\eps = \sum_t \eps_t~~~~\textrm{and}~~~~\delta = \sum_t \delta_t$$
\item If $M_1,\dots,M_T$ are $(\eps_0,\delta_1),\dots,(\eps_0,\delta_T)$-DP for some $\eps_0 \leq 1$, then for every $\delta_0 > 0$, $M$ is $(\eps, \delta)$-DP for
$$\eps = \eps_0 \sqrt{6 T \log(1/\delta_0)}~~~~\textrm{and}~~~~\delta = \delta_0 + \sum_t \delta_t$$
\item If $M_1,\dots,M_T$ are $\rho_1,\dots,\rho_T$-zCDP then $M$ is $\rho$-zCDP for $\rho = \sum_t \rho_t$.
\end{enumerate}
\end{lem}

Note that the first and the third properties say that $(\eps,\delta)$-DP and $\rho$-zCDP compose linearly---the parameters simply add up.  The second property says that $(\eps,\delta)$-DP actually composes sublinearly---the parameter $\eps$ grows roughly with the square root of the number of steps in the composition, provided we allow a small increase in $\delta$.

\medskip
\noindent\textbf{The Gaussian Mechanism.}
Our algorithms will extensively use the well known and standard Gaussian mechanism to ensure differential privacy.

\begin{defn}[$\ell_2$-Sensitivity]
Let $f : \cX^n \to \R^d$ be a function, its \emph{$\ell_2$-sensitivity} is
$$
\Delta_{f} = \max_{X \sim X' \in \cX^n} \| f(X) - f(X') \|_{2}
$$
\end{defn}

\begin{lem}[Gaussian Mechanism] \label{lem:gaussiandp}
Let $f : \cX^n \to \R^d$ be a function with $\ell_2$-sensitivity $\Delta_{f}$.  Then the Gaussian mechanism
$$M_{f}(X) = f(X) + N\left(0, \left(\frac{\Delta_{f}}{\sqrt{2\rho}}\right)^2 \cdot \id\right)$$
satisfies $\rho$-zCDP.
\end{lem}

\noindent In order to prove accuracy, we will use the following standard tail bounds for Gaussian random variables.

\begin{lem}\label{fact:gaussian-error}
    If $Z \sim N(0,\sigma^2)$ then for every $t > 0$,
    $\pr{}{|Z| > t\sigma} \leq 2e^{-t^2/2}.$
\end{lem}

\noindent\textbf{Parameter Estimation to Distribution Estimation.}
In this work, our goal is to estimate some underlying distribution in total variation distance.
For both Gaussian and product distributions, we will achieve this by estimating the parameters of the distribution, and we argue that a distribution from the class with said parameters will be accurate in total variation distance.
For product distributions, we require an estimate of the parameters, which is accurate in terms of a type of chi-squared distance; this is shown in the proof of Theorem~\ref{thm:ppde_acc}.
For Gaussian distributions, the parameter estimate we require is slightly more difficult to describe.
For a vector $x$, define $\|x\|_\Sigma = \|\Sigma^{-1/2}x\|_2$.
Similarly, for a matrix $X$, define $\|X\|_\Sigma = \|\Sigma^{-1/2}X\Sigma^{-1/2}\|_F$.
With these two norms in place, we have the following lemma, which is a combination of Corollaries 2.13 and 2.14 of~\cite{DiakonikolasKKLMS16}.
\begin{lemma}
\label{lem:gaussian-tv}
Let $\alpha \geq 0$ be smaller than some absolute constant.
Suppose that $\|\mu - \hat \mu\|_\Sigma \leq \alpha$, and $\|\Sigma - \hat \Sigma\|_\Sigma \leq \alpha$, where $\mathcal{N}(\mu, \Sigma)$ is a Gaussian distribution in $\mathbb{R}^d$, $\hat \mu \in \mathbb{R}^d$, and $\Sigma \in \mathbb{R}^{d \times d}$ is a PSD matrix.
Then $\dtv(\mathcal{N}(\mu, \Sigma), \mathcal{N}(\hat \mu, \hat \Sigma)) \leq O(\alpha)$.
\end{lemma}


\section{Private Covariance Estimation for Gaussians}
\label{sec:gaussian-cov}
In this section we present our algorithm for privately estimating the covariance of an unknown Gaussian.  Suppose we are given i.i.d.\ samples $X_1, \ldots, X_n \sim \normal (0, \Sigma)$ where $\id \preceq \Sigma \preceq \kappa \id$.  Our goal is to privately output $\wh\Sigma$ so that
$$
\| \Sigma - \wh\Sigma \|_{\Sigma} \leq O(\alpha),
$$ 
where $\| A \|_\Sigma = \| \Sigma^{-1/2} A \Sigma^{-1/2} \|_F$.  Here the matrix square root denotes any possible square root; it is trivial to check that all such choices are equivalent.  By Lemma~\ref{lem:gaussian-tv}, this condition ensures 
\[
\dtv (\normal (0, \Sigma), \normal (0, \Sigmahat)) \leq O(\alpha).
\]

\subsection{Useful Concentration Inequalities}
We will need several facts about Gaussians and Gaussian matrices.  Throughout this section, let $\GUE(\sigma^2)$ denote the distribution over $d \times d$ symmetric matrices $M$ where for all $i \leq j$, we have $M_{ij} \sim \normal (0, \sigma^2)$ i.i.d..
From basic random matrix theory, we have the following guarantee.
\begin{theorem}[see e.g. \cite{Tao12} Corollary 2.3.6]
\label{thm:GUE}
For $d$ sufficiently large, there exist absolute constants $C, c > 0$ such that
\[
\pr{M \sim \GUE (\sigma^2) }{\| M \|_{2} > A \sigma \sqrt{d}} \leq C \exp (-c A d)
\]
for all $A \geq C$.
\end{theorem}
\noindent
We also require the following, well known tail bound on quadratic forms on Gaussians.
\begin{theorem}[Hanson-Wright Inequality (see e.g.~\cite{LaurentM00})]
\label{thm:hanson-wright}
Let $X \sim \normal (0, \id)$ and let $A$ be a $d \times d$ matrix.
Then, for all $t > 0$, the following two bounds hold:
\begin{align}
&\pr{}{X^\top A X - \tr (A) \geq 2 \| A \|_F \sqrt{t} + 2 \| A \|_2 t} \leq \exp (-t) \label{eq:hw-ub} \\
&\pr{}{X^\top A X - \tr (A) \leq -2 \| A \|_F \sqrt{t}} \leq \exp (-t)\label{eq:hw-lb}
\end{align}
\end{theorem}
\noindent As a special case of the above inequality, we also have
\begin{fact}[\cite{LaurentM00}]
\label{fact:chi-squared}
Fix $\beta > 0$, and let $X_1, \ldots, X_m \sim \normal (0, \sigma^2)$ be independent.
Then
\[
\Pr \left[ \left| \frac{1}{m} \sum_{i = 1}^m X_i^2 - \sigma^2 \right| > 4 \sigma^2 \left( \sqrt{\frac{\log(1/\beta)}{m}} + \frac{2 \log(1/\beta)}{m} \right) \right] \leq \beta
\]
\end{fact}
\ifnum\epsdeltastuff=1
Another simple corollary of this fact is the following.
\begin{coro}
\label{cor:hanson-wright-extreme-lb}
Let $X \sim \normal (0, 1) \in \R^d$ and let $A$ be a $d \times d$ PSD matrix. 
Then, there exists a universal constant $C \leq 10000$ so that
\[
\Pr \left[ X^\top A X < \frac{1}{C} \tr (A) \right] < \frac{1}{10} \; .
\]
\end{coro}

\noindent The proof of this corollary appears in Section~\ref{sec:hanson-wright-extreme-lb}.
\fi

\subsection{Deterministic Regularity Conditions}
We will rely on certain regularity properties of i.i.d.\ samples from a Gaussian.
These are standard concentration inequalities, and a reference for these facts is Section 4 of~\cite{DiakonikolasKKLMS16}.
\begin{fact}
\label{fact:gaussian-facts}
Let $X_1, \ldots, X_n \sim \normal (0, \Sigma)$ i.i.d.\ for $\id \preceq \Sigma \preceq \kappa \id$.  Let $Y_i = \Sigma^{-1/2} X_i$ and let 
\[
\Sigmahat_Y = \frac{1}{n} \sum_{i = 1}^n Y_i Y_i^\top
\]
Then for every $\beta > 0$, the following conditions hold except with probability $1-O(\beta)$.
\begin{align}
&\forall i \in [n]~~~\| Y_i \|_2^2 \leq O\left( d \log (n / \beta) \right) \label{eq:cov-cond1} \\
&\left( 1 - O \left( \sqrt{\frac{d + \log(1/\beta)}{n}} \right) \right) \cdot \id \preceq  \Sigmahat_Y \preceq \left( 1 + O \left( \sqrt{\frac{d + \log(1/\beta)}{n}} \right) \right) \cdot \id \label{eq:cov-cond2} \\
&\left\| \id - \Sigmahat_Y \right\|_F \leq O \left( \sqrt{\frac{d^2 + \log(1/\beta)}{n}} \right) \label{eq:cov-cond3}
\end{align}
\end{fact}

\noindent We now note some simple consequences of these conditions.
These inequalities follow from simple linear algebra and we omit their proof for conciseness.
\begin{lemma}
\label{lem:rotations-cov-conds}
Let $Y_1, \ldots, Y_n$ satisfy \eqref{eq:cov-cond1}--\eqref{eq:cov-cond3}.
Fix $M \succ 0$, and for all $i = 1, \ldots, n$, let $Z_i = M^{1/2} Y_i$, and let $\wh\Sigma_{Z} = \frac{1}{n} \sum_{i=1}^{n} Z_i Z_i^\top$.  Let $\kappa'$ be the top eigenvalue of $M$.
Then
\begin{align*}
&\forall i \in [n]~~~\| Z_i \|_2^2 \leq O\left(\kappa' d \log (n / \beta) \right) \\
&\left( 1 - O \left( \sqrt{\frac{d + \log(1/\beta)}{n}} \right) \right) \cdot M \preceq \wh\Sigma_{Z} \preceq \left( 1 + O \left( \sqrt{\frac{d + \log(1/\beta)}{n}} \right) \right) \cdot M \\
&\left\| M - \wh\Sigma_Z \right\|_M \leq O \left( \sqrt{\frac{d^2 + \log(1/\beta)}{n}} \right)
\end{align*}
\end{lemma}

\subsection{A Simple Algorithm for Well Conditioned Gaussians}
\noindent 
We first consider the following simple algorithm:~remove all points whose norm exceeds a certain threshold, then compute the empirical covariance of the resulting data set, and perturb the empirical covariance with noise to preserve privacy. This algorithm will have nearly-optimal dependence on most parameters, however, it will have a polynomial dependence on the condition number.  Pseudocode for this algorithm is given in Algorithm \ref{alg:naivePCE}.

\begin{algorithm}[h!] \label{alg:naivePCE}
\caption{Naive Private Gaussian Covariance Estimation $\NaivePCE_{\rho,\beta, \kappa}(X)$}
\KwIn{A set of $n$ samples $X_1, \ldots, X_n$ from an unknown Gaussian.  Parameters $\rho, \beta, \kappa > 0$}
\KwOut{A covariance matrix $M$.} \vspace{10pt}

Let 
$
S \gets \left\{ i \in [n]: \| X_i \|_2^2 \leq O(d \kappa \log (n / \beta)) \right\}
$\\
Let 
$$
\sigma \gets \Theta \left( \frac{d \kappa \log (\frac{n}{ \beta})}{n\rho^{1/2}} \right)
$$

Let $M' \gets \frac{1}{n} \sum_{i \in S} X_i X_i^\top + N$ where $N \sim \GUE (\sigma^2)$ \\
Let $M$ be the projection of $M'$ into the set of PSD matrices.

\Return $M$
\end{algorithm}

\begin{lemma} [Analysis of \NaivePCE]
\label{lem:dp-naive}
For every $\rho, \beta, \kappa, n$, $\NaivePCE_{\rho,\beta,\kappa}(X)$ satisfies $\rho$-zCDP, and if $X_1,\dots,X_n$ are sampled i.i.d.\ from $\normal(0,\Sigma)$ for $\id \preceq \Sigma \preceq \kappa \id$ and satisfy \eqref{eq:cov-cond1}--\eqref{eq:cov-cond3}, then with probability at least $1-O(\beta)$, it outputs $M$ so that $M = \Sigma^{1/2} (I + N') \Sigma^{1/2} + N$ where $\Sigma^{1/2} (I + N') \Sigma^{1/2} 
\succeq 0$, and
\begin{align}
  \| N' \|_F &\leq O \left( \sqrt{\frac{d^2 + \log(1/\beta)}{n}} \right) ~\mathrm{and}~ \| N \|_\Sigma \leq O \left( \frac{d^{2} \kappa \log (n/\beta)  \log^{1/2}(1/\beta)}{n \rho^{1/2}} \right) \; ,~\mathrm{and} \label{eq:dp-naive-frob} \\
\| N' \|_2 &\leq O \left( \sqrt{\frac{d + \log(1/\beta)}{n}} \right) ~\mathrm{and}~ \| N \|_2 \leq O \left( \frac{d^{3/2} \kappa \log (n/\beta) \log(1/\beta)}{n \rho^{1/2}} \right) \; . \label{eq:dp-naive-op}
\end{align}

\end{lemma}

\begin{proof}
We first prove the privacy guarantee.
Given two neighboring data sets $X, X'$ of size $n$ which differ in that one contains $X_i$ and the other contains $X'_i$, the truncated empirical covariance of these two data sets can change in Frobenius norm by at most
\[
\left\| \frac{1}{n} \left(X_iX_i^\top - X'_i (X'_i)^\top \right) \right\|_F \leq \frac{1}{n} \| X_i \|_2^2 + \frac{1}{n} \| X_i' \|_2^2 \leq O \left( \frac{d \kappa \log (n/\beta) }{n} \right) \; .
\]
\noindent
Thus the privacy guarantee follows immediately from Lemma~\ref{lem:gaussiandp}.

We now prove correctness.
Recall $M'$ is the original noised covariance before projection back to the SDP cone.
We first prove that $M'$ satisfies this form, with $\Sigma^{1/2} (I + N') \Sigma^{1/2} = \frac{1}{n} \sum_{i = 1}^n X_i X_i^\top$.
Clearly this implies that $\Sigma^{1/2} (I + N') \Sigma^{1/2} \succeq 0$.
Since \eqref{eq:cov-cond1} holds, we have $S = [n]$.  The first inequality in \eqref{eq:dp-naive-frob} now follows from Lemma~\ref{lem:rotations-cov-conds}, and the second follows from Fact~\ref{fact:chi-squared} and since $\| N \|_\Sigma \leq \| N \|_F$, as $\Sigma \succeq I$.  By a similar logic, \eqref{eq:dp-naive-op} holds since we can apply Lemma~\ref{lem:rotations-cov-conds} and Theorem~\ref{thm:GUE}.
Finally, to argue about $M$, simply observe that $\Sigma^{1/2} (I + N') \Sigma^{1/2}$ is PSD, and projection onto the PSD cone can only decrease distance (in either spectral norm or Frobenius norm) to any element in the PSD cone.
\end{proof}

The following is an immediate consequence of Lemma~\ref{lem:dp-naive}, seen by noting that $\|\Sigma - M\|_\Sigma = \|\Sigma^{-1/2} N \Sigma^{-1/2} + N'\|_F \leq \|N\|_\Sigma + \|N'\|_F$.
\begin{thm} \label{thm:naivepce}
For every $\rho, \beta, \kappa > 0$, the algorithm $\NaivePCE_{\rho,\beta,\kappa}$ is $\rho$-zCDP and, when given
$$
n = O\left( \frac{d^2 + \log(\frac{1}{\beta})}{\alpha^2} + \frac{\kappa d^2 \polylog(\frac{\kappa d}{\alpha \beta \rho})}{\alpha \rho^{1/2}} \right), 
$$ 
samples from $\normal(0,\Sigma)$ satisfying $\id \preceq \Sigma \preceq \kappa \id$, with probability at least $1- O(\beta)$, it returns $M$ such that $\| \Sigma - M \|_{\Sigma} \leq O(\alpha).$
\end{thm}

\subsection{A Private Recursive Preconditioner}

When $\kappa$ is a constant, Theorem~\ref{thm:naivepce} says that \NaivePCE~privately estimates the covariance of a Gaussian with little overhead compared to non-private estimation.  
In this section we will show how to nearly eliminate the dependence on the covariance by privately learning a \emph{preconditioner} $A$ such that $\id \preceq A \Sigma A \preceq 1000 \id$.  
Once we have this preconditioner, we can reduce the condition number of the distribution to a constant.
In this state, we can apply \NaivePCE~to estimate the covariance at no cost in $\kappa$.

\subsubsection{Reducing the Condition Number by a Constant}

Our preconditioner works recursively.  The main ingredient in the recursive construction is an algorithm, \PPreCond{} (Algorithm~\ref{alg:private-precond}) that privately estimates a matrix $A$ such that the condition number of $A\Sigma A$ improves over that of $\Sigma$ by a constant factor.  Once we have this primitive we can apply it recursively in a straightforward way.  Note that in Algorithm~\ref{alg:private-precond}, when we apply \NaivePCE~to obtain a weak estimate of $\Sigma$, we use too few samples for \NaivePCE~to obtain a good estimate of $\Sigma$ on its own.
\begin{algorithm}[h!] \label{alg:private-precond}
\caption{Private Preconditioning $\PPreCond_{\rho,\beta,\kappa,K}(X)$}
\KwIn{A set of $n$ samples $X_1, \ldots, X_n$ from an unknown Gaussian.  Parameters $\rho, \beta, \kappa, K > 0$.}
\KwOut{A symmetric matrix $A$.}\vspace{10pt}
Let $Z \gets \NaivePCE_{\rho,\beta,\kappa}(X_1, \ldots, X_n)$\\
Let $(\lambda_1,v_1), \dots, (\lambda_d,v_d)$ be the eigenvalues and the corresponding eigenvectors of $Z$\\
Let $V \gets \mathrm{span} \left( \left\{v_i : \lambda_i \geq \frac{\kappa}{2} \right\} \right) \subseteq \R^d$

\vspace{10pt}
\Return the pair $(V,A)$ where 
$$A = \frac{1}{\sqrt{K}} \Pi_{V} + \Pi_{V^\perp}$$
\end{algorithm}
\noindent

The guarantee of Algorithm~\ref{alg:private-precond} is captured in the following theorem.
\begin{thm}
\label{thm:dp-precondition}
For every $\rho,\beta,\kappa,K > 0$, $\PPreCond_{\rho,\beta,\kappa,K}(X)$ satisfies $\rho$-zCDP and, if $X_1,\dots,X_n$ are sampled i.i.d.\ from $\normal(0,\Sigma)$ for $\id \preceq \Sigma \preceq \kappa \id$ and satisfy \eqref{eq:cov-cond1}--\eqref{eq:cov-cond3}, then with probability at least $1 - O(\beta)$ it outputs $(V,A)$ such that
\begin{equation}
\label{eq:dp-preconditioning}
(1 - \psi)^2 (1 - \Gamma) \id \preceq A \Sigma A \preceq (1 + \psi) \cdot \kappa \left( \max\left( \frac{1}{K}, \frac{1}{2} \right) + \varphi \right) \id
\end{equation}
where
$$
\varphi = O \left( \frac{d^{3/2} \log (n / \beta) \log (1 / \beta)}{n \rho^{1/2}} \right),
$$
$$
\psi = O \left( \sqrt{\frac{d + \log(1 / \beta)}{n}} \right),\mbox{ and }
$$
$$
\Gamma = \max \left\{ \frac{2 K}{(1/2 - \varphi) \kappa}~, \frac{16 K \varphi^2}{(1/2 - \varphi)^2} \right\}$$
In particular, if $\kappa > 1000$, and $K$ is an appropriate constant, and
$$n \geq O\left(\frac{d^{3/2} \polylog(\frac{d}{\rho \beta})}{\rho^{1/2}} \right)$$
then $1.1A$ is such that $\id \preceq (1.1A) \Sigma (1.1A) \preceq \frac{7}{10} \kappa \id$.
\end{thm}

\begin{proof}
Privacy follows since we are simply post-processing the output of Algorithm~\ref{alg:naivePCE} (Lemma~\ref{lem:postprocessing}).
Thus it suffices to prove correctness.  We assume that \eqref{eq:cov-cond1}--\eqref{eq:cov-cond3} hold simultaneously.  By Lemma~\ref{lem:dp-naive}, \eqref{eq:dp-naive-frob}--\eqref{eq:dp-naive-op} hold simultaneously for the matrix $Z$ except with probability $O(\beta)$.  We will condition on these events throughout the remainder of the proof.  Observe that \eqref{eq:cov-cond2} implies that $\Sigmahat = \frac{1}{n} \sum_{i = 1}^n X_i X_i^\top$ is non-singular.

We will prove the upper bound and lower bound in \eqref{eq:dp-preconditioning} in two separate lemmata.

\begin{lemma}
\label{lem:upper-bound}
Let $V,A$ be as in Algorithm~\ref{alg:private-precond}.
Then, conditioned on \eqref{eq:cov-cond1}--\eqref{eq:dp-naive-op}, with probability $1 - O(\beta)$, we have
\[
 \left\| A \Sigma A \right\|_2 \leq (1 + \psi) \cdot \kappa \left( \max\left( \frac{1}{K}, \frac{1}{2} \right) + \varphi \right) \; .
\]
 \end{lemma}

\begin{lemma}
\label{lem:lower-bound}
Let $V,A$ be as in Algorithm~\ref{alg:private-precond}.
Then, conditioned on \eqref{eq:cov-cond1}--\eqref{eq:dp-naive-op}, with probability $1 - O(\beta)$, we have
\begin{equation}
\label{eq:lower-bound}
A \Sigma A \succeq (1 - \psi)^2 (1 - \Gamma) \id \; .
\end{equation}
 \end{lemma}
 \noindent
These two lemmata therefore together imply Theorem~\ref{thm:dp-precondition}.
We now turn our attention to the proofs of these lemmata.
Let $N$ be the Gaussian noise added to the empirical covariance in $\NaivePCE$, so that $Z = \Sigmahat + N$.

\begin{proof}[Proof of Lemma~\ref{lem:upper-bound}]
By Lemma~\ref{lem:rotations-cov-conds} (with $M = \Sigma$), it suffices to show that 
\[
\| A \Sigmahat A \|_2 \leq \kappa \left( \max\left( \frac{1}{K}, \frac{1}{2} \right) + \varphi \right) \; .
\]
But with probability $1 - \beta$, we have 
\begin{align*}
\left\| A \Sigmahat A \right\|_2 &\leq \left\| A Z A \right\|_2 + \left\| A N A \right\|_2 \\
&\stackrel{(a)}{\leq} \left\| A Z A \right\|_2 + \kappa \varphi \; ,
\end{align*}
where (a) follows since $\| A \|_2 \leq 1$ and Theorem~\ref{thm:GUE}.
We now observe that since $V$ is a span of eigenvectors of $Z$, we have
\[
A Z A = \frac{1}{K} \Pi_V Z \Pi_V + \Pi_{V^\perp} Z \Pi_{V^\perp} \; ,
\]
and so by our choice of $V$, we have $\| A Z A \|_2 \leq \kappa \cdot \max (1 / K, 1/2)$.
This completes our proof.
\end{proof}

\noindent We now prove the lower bound in Theorem \ref{thm:dp-precondition}:
\begin{proof}[Proof of Lemma~\ref{lem:lower-bound}]
As before, by Lemma~\ref{lem:rotations-cov-conds}, it suffices to prove that
\[
A \Sigmahat A \succeq (1 - \psi) \left( 1 - \Gamma \right) \id \; .
\]
This is equivalent to showing that for all unit vectors $u$, we have 
\[
u^T A \Sigmahat A u \geq (1 - \psi) \left( 1 - \Gamma \right) \; .
\]
Fix any such $u$.
Expanding, we have
\begin{equation}
u^T A \Sigmahat A u = \frac{1}{K} u^T \Pi_{V} \Sigmahat \Pi_V u + \frac{1}{K^{1/2}} u^T \Pi_{V} \Sigmahat \Pi_{V^\perp} u + \frac{1}{K^{1/2}} u^T \Pi_{V^\perp} \Sigmahat \Pi_{V} u + u^T \Pi_{V^\perp} \Sigmahat \Pi_{V^\perp} u \; . \label{eq:lb-expanded}
\end{equation}
The first and last terms are non-negative since $\Sigmahat$ is PSD, but the other terms may be negative, so we need to control their magnitude.
Note that
\begin{align*}
u^T \Pi_{V} \Sigmahat \Pi_V u &= u^T \Pi_V Z \Pi_V u - u^T \Pi_V N \Pi_V u  \\
&\geq \frac{\kappa}{2} \| \Pi_V u \|_2^2 - \kappa \varphi \| \Pi_V u \|_2^2 = \kappa \left( \frac{1}{2} - \varphi \right) \left\| \Pi_V u \right\|_2^2 \; .
\end{align*}
where the inequality follows from our choice of $V$ (the ``large'' directions of $Z$), and Theorem~\ref{thm:GUE} (bounding the spectral norm of $N$).
On the other hand, we have
\begin{align*}
\left| \frac{1}{K^{1/2}} u^T \Pi_{V} \Sigmahat \Pi_{V^\perp} u \right| &= \left| \frac{1}{K^{1/2}} u^T \Pi_{V} (Z - N) \Pi_{V^\perp} u \right| \\
&\stackrel{(a)}{=} \left| \frac{1}{K^{1/2}} u^T \Pi_V N \Pi_{V^\perp} u \right| \\
&\stackrel{(b)}{\leq} \frac{\kappa}{K^{1/2}} \varphi \| \Pi_V u \|_2 \| \Pi_{V^\perp} u \|_2 \\
&\leq \frac{\kappa}{K^{1/2}} \varphi \| \Pi_V u \|_2 \; ,
\end{align*}
where (a) follows since $\Pi_V Z \Pi_{V^\perp} = 0$, and (b) follows from Theorem~\ref{thm:GUE}.
Similarly, we have
\[
\left| \frac{1}{K^{1/2}} u^T \Pi_{V^\perp} \Sigmahat \Pi_{V} u \right| \leq \frac{\kappa}{K^{1/2}} \varphi \| \Pi_V u \|_2 \; .
\]
Thus, if we have $\| \Pi_V u \|_2^2 \geq \Gamma$, by our choice of $\Gamma$, we have
\begin{align*}
&\frac{1}{K} u^T \Pi_{V} \Sigmahat \Pi_V u + \frac{1}{K^{1/2}} u^T \Pi_{V} \Sigmahat \Pi_{V^\perp} u + \frac{1}{K^{1/2}} u^T \Pi_{V^\perp} \Sigmahat \Pi_{V} u \\
\geq~&\frac{\kappa}{K} \left( \frac{1}{2} - \varphi \right) \left\| \Pi_V u \right\|_2^2 - 2 \frac{\kappa}{K^{1/2}} \varphi \| \Pi_V u \|_2 \\
\geq~&\frac{\kappa}{2 K} \left( \frac{1}{2} - \varphi \right) \left\| \Pi_V u \right\|_2^2 \\
\geq~&1 \; .
\end{align*}
Thus in this case the claim follows since the final term in~\eqref{eq:lb-expanded} is nonnegative since $\Sigmahat$ is PSD.

Now consider the case where
\[
\| \Pi_V u \|_2 < \Gamma \; ,
\]
or equivalently, since by the Pythagorean theorem we have $\| \Pi_V u \|_2^2  + \| \Pi_{V^\perp} u \|_2^2 = 1$, 
\[
\| \Pi_{V^\perp} u \|_2^2 > 1 - \Gamma \; .
\]
Then, since we have $\Sigmahat \succeq (1 - \psi) \id$ (Fact~\ref{fact:gaussian-facts}), we have
\begin{align*}
u^T A \Sigmahat A u &\geq (1 - \psi) u^T \left( \frac{1}{K} \Pi_V \Pi_V + \Pi_{V^\perp} \Pi_{V^\perp} \right) u \\
& \geq (1 - \psi) \| \Pi_{V^\perp} u \|_2^2 \\
&\geq (1 - \psi) \left( 1 - \Gamma \right) \; ,
\end{align*}
as claimed.
\end{proof}
\noindent
Combining Lemma~\ref{lem:upper-bound} and Lemma~\ref{lem:lower-bound} yield the desired conclusion.
\end{proof}

\subsubsection{Recursive Preconditioning}

Once we have \PPreCond, we can apply it recursively to obtain a private preconditioner, \PPC~(Algorithm~\ref{alg:PPC}) that reduces the condition number down to a constant.
\begin{algorithm}[h!] \label{alg:PPC}
\caption{Privately estimating covariance $\PPC_{\rho,\beta,\kappa}(X)$}
\KwIn{A set of $n$ samples $X_1, \ldots, X_n$ from an unknown Gaussian $\normal (0, \Sigma)$.  Parameters $\rho,\beta, \kappa > 0$}
\KwOut{A symmetric matrix $A$} \vspace{10pt}

Let  $T \gets O(\log \kappa)~~~\rho' \gets \rho/T~~~\beta' \gets \beta/T$\\
Let $K$ be the constant from Theorem~\ref{thm:dp-precondition}.\\
Let $\kappa^{(1)} \gets \kappa$ and let $X^{(1)}_i \gets X_i$ for $i = 1, \ldots, n$.\\
\For{$t = 1, \ldots, T$}{
Let $\tilde{A}^{(t)} \gets \PPreCond_{\rho', \beta', \kappa^{(t)},K}(X^{(t)}_1, \ldots, X^{(t)}_n)$, and let $A^{(t)} \gets 1.1 \tilde{A}^{(t)}$.\\
Let $\kappa^{(t+1)} \gets 0.7 \kappa^{(t)}$\\
Let $X^{(t+1)}_i \gets A^{(t)} X^{(t)}_i$ for $i = 1, \ldots, n$.}

\vspace{10pt}
\Return{The matrix $A = \prod_{t=1}^{T} A^{(t)}$}.
\end{algorithm}

\begin{theorem}
\label{thm:cov-kappa}
For every $\rho, \alpha, \beta,\kappa > 0$, $\PPC_{\rho,\beta,\kappa}$ satisfies $\rho$-zCDP, and when given
\[
n = \Omega\left( \frac{d^{3/2} \log^{1/2}(\kappa) \polylog(\frac{d \log \kappa}{
\rho \beta}) }{\rho^{1/2}} \right)
\]
samples $X_1,\dots,X_n \sim \normal(0,\Sigma)$ for $\id \preceq \Sigma \preceq \kappa \id$, with probability $1-O(\beta)$ it outputs a symmetric matrix $A$ such that $\id \preceq A\Sigma A \preceq 1000 \id$.
\end{theorem}

\begin{proof}
Privacy is immediate from Theorem~\ref{thm:dp-precondition} and composition of $\rho$-zCDP (Lemma~\ref{lem:dpcomp}).

By Fact~\ref{fact:gaussian-facts}, \eqref{eq:cov-cond1}--\eqref{eq:cov-cond3} hold for the sample $X_1,\dots,X_n$ except with probability $O(\beta)$.  Define $\Sigma^{(1)} = \Sigma$ and recursively define $\Sigma^{(t)} = A^{(t-1)} \Sigma^{(t-1)} A^{(t-1)}$ to be the covariance after the $t$-th round of preconditioning.
    By the guarantee of $\PPreCond$ (Theorem~\ref{thm:dp-precondition}), a union bound, and our choice of $n$, we have that for every $t$, we obtain a matrix $A^{(t)}$ such that $$\id \preceq A^{(t)} \Sigma^{(t)} A^{(t)} \preceq 0.7 \kappa^{(t)} \id.$$
The theorem now follows by induction on $t$.
\end{proof}

\subsection{Putting It All Together}

We can now combine our private preconditioning algorithm with the na\"ive algorithm for covariance estimation to obtain a complete algorithm for covariance estimation.
\begin{algorithm}[h!] \label{alg:privatecovariance}
\caption{Private Covariance Estimator $\PGCEKappa_{\rho, \beta, \kappa}(X)$}
\KwIn{Samples $X_1,\dots,X_{n}$ from an unknown Gaussian $\normal(0, \Sigma)$.  Parameters $\rho,\beta, \kappa > 0$.}
\KwOut{A matrix $\wh{\Sigma}$ such that $\| \Sigma - \wh{\Sigma} \|_\Sigma \leq \alpha$.}
\vspace{10pt}
Let $\rho' \gets \rho/2$ and $\beta' \gets \beta/2$\\
Let $A \gets \PPC_{\rho', \beta', \kappa}(X_1,\dots,X_n)$ be the private preconditioner \\\vspace{10pt}
Let $Y_i \gets A X_i$ for $i = 1,\dots,n$.\\
Let $\wt\Sigma \gets \NaivePCE_{\rho', \beta', 1000}(Y_1,\dots,Y_n)$\vspace{10pt}

\Return $\wh\Sigma = A^{-1} \wt\Sigma A^{-1}$
\end{algorithm}

This algorithm has the following guarantee.
\begin{theorem}
\label{thm:pce-kappa}
For every $\rho,\beta,\kappa > 0$, $\PGCEKappa_{\rho.\beta,\kappa}(X)$ is $\rho$-zCDP and, when given
$$
n = O\left( \frac{d^{2} + \log(\frac{1}{\beta})}{\alpha^2} + \frac{d^2 \polylog(\frac{d}{\alpha \beta \rho})}{\alpha \rho^{1/2}} + \frac{d^{3/2}  \log^{1/2}(\kappa) \polylog(\frac{d \log \kappa}{
\rho \beta})}{\rho^{1/2}} \right)
$$
$X_1,\dots,X_n \sim \normal(0,\Sigma)$ for $\id \preceq \Sigma \preceq \kappa \id$, with probability $1 - O(\beta)$, it outputs $\wh\Sigma$ such that $\| \Sigma  - \wh\Sigma \|_{\Sigma} \leq O(\alpha)$.
\end{theorem}
\begin{proof}
Privacy follows from Theorem~\ref{thm:dp-precondition} and composition of $\rho$-zCDP (Lemma~\ref{lem:dpcomp}).

By construction, the samples $Y_1,\dots,Y_n$ are i.i.d.\ from $\normal(0, A \Sigma A)$.  By Theorem~\ref{thm:cov-kappa}, and our choice of $n$, we have that except with probability $O(\beta)$, $A$ is such that $\id \preceq A\Sigma A \preceq 1000 \id$.  Therefore, combining the guarantees of \NaivePCE~with our choice of $n$ we obtain that, except with probability $O(\beta)$, $\wt\Sigma$ satisfies $\| A \Sigma A - \wt\Sigma \|_{A \Sigma A} \leq O(\alpha)$.  The theorem now follows because $\| \Sigma - \wh \Sigma \|_{\Sigma} = \| \Sigma - A^{-1} \wt \Sigma A^{-1} \|_{\Sigma} = \| A \Sigma A - \wt \Sigma \|_{A \Sigma A} \leq O(\alpha).$
\end{proof}

\ifnum\epsdeltastuff=1
\section{Condition-Number-Free Preconditioning} \label{sec:gaussian-cov-nokappa}

In this section we build on the results of
Section~\ref{sec:gaussian-cov} to obtain a
private preconditioner for Gaussians that
have potentially \emph{unbounded} covariance.
That is, we require $\Sigma \succeq \id$ but
do not require any upper bound on $\| \Sigma \|_2$.
Note that, since we don't need an upper bound
on $\Sigma$, the lower bound $\Sigma \succeq \id$
is arbitrary, and can be replaced with
$\Sigma \succeq \gamma \id$ for any $\gamma > 0$.

Intuitively, the polylogarithmic dependence on $\kappa$ in Theorem~\ref{thm:cov-kappa} comes from the fact that the recursive preconditioner is oblivious to the actual spectrum of $\Sigma$, and only reduces the upper bound on the condition number by a constant factor in each round, so we have to precondition for $\Omega(\log \kappa)$ rounds.  At a high-level, our algorithm for estimating a Gaussian with unbounded covariance starts by obtaining a rough estimate of $\| \Sigma \|_2$ and uses this information to reduce the condition number much faster, so that the number of rounds of recursive preconditioning is independent of $\kappa$.  In order to avoid any dependence on $\kappa$ in the sample complexity, our algorithm uses the flexibility of $(\eps,\delta)$-differential privacy (for $\delta > 0$) in an essential way.


\newcommand{\indi}{\mathbf{1}}

\subsection{Estimating the Spectral Norm}

First we describe how to obtain a rough estimate of the spectral norm of the covariance matrix $\Sigma$.  Our algorithm makes use of an algorithm for privately computing histograms~\cite{BeimelNS13,Vadhan17}.  Given a \emph{universe} $\cX$ and a dataset $X = (X_1,\dots,X_n) \in \cX^n$, the \emph{histogram} is the vector
\[
q(X) = \left(\frac{1}{n} \sum_{i = 1}^n \indi_{y} (X_i) \right)_{y \in \cX}
\]
where $\indi_{y}(x)$ is the function that returns $1$ if and only if $x = y$.
That is, a point query simply asks what fraction of the points in the data set are equal to a single point in the universe.
\begin{theorem}[\cite{BeimelNS13,Vadhan17}]
\label{thm:private-histogram}
For every universe $\cX$, $n \in \N$, $0 < \eps,\delta,\beta \leq 1$,
there is a polynomial time $(\eps, \delta)$-DP mechanism $H : \cX^n \to \R^\cX$ such that, for every $X \in \cX^n$, with probability at least $1 - O(\beta)$,
$$
\left\| q(X) - H(X) \right\|_{\infty} \leq O\left( \frac{\log(\frac{n}{\delta\beta})}{\eps n} \right)
$$
\end{theorem}

With this theorem in hand, we can now describe our algorithm.  To describe the algorithm, we need to describe a particular set of point queries.  We will use these definitions throughout.
Let $C > 0$ be a sufficiently large universal constant to be specified later.  For any $r \in \R_{\geq 0}$, define the interval $I_r = [C^{r-1}, C^r]$.  Define the infinite domain to be a partition of $[\frac{1}{C^2},\infty)$ into geometrically increasing intervals $I_r$.  Specifically,
$$
\cX = \left\{~I_{\log_C(d) - 1}~,~I_{\log_C(d)}~,~I_{\log_C(d) + 1}~,~\dots~\right\} \cup \bot
$$
We then bucket the samples $X_1,\dots,X_n$ by their $\ell_2^2$ norm.  Specifically, for each sample $X_i \in \R^d$ we let $B_i = I_r$ if $\| X_i \|_2^2 \in I_r$ for some $I_r \in \cX$ and $B_i = \bot$ otherwise.  We now construct the dataset $B = (B_1,\dots,B_n)$ and run $H(B)$ to find an interval $I_r$ such that at least $n/4$ samples lie in $I_r$.
The pseudocode for this algorithm is given in Algorithm~\ref{alg:estimate-trace}.
\begin{algorithm}[h!] \label{alg:estimate-trace}
\caption{Privately estimating spectral norm $\PEstimateTrace_{\eps,\delta,\beta}(X_1, \ldots, X_n)$}
\KwIn{A set of $n$ samples from an unknown Gaussian $\normal (0, \Sigma)$.  Parameters $\eps, \delta, \beta > 0$.}
\KwOut{An estimate of $\| \Sigma \|_2$}\vspace{10pt}

Let $C, I_r, \cX, B$ and $H$ be as described above and let $(h_{I_{r}})_{I_r \in \cX} \gets H(B_1,\dots,B_n)$\\
If possible, let $r$ be such that $h_{I_{r}} \geq 1/4$\\\vspace{10pt}

\Return $C^r$ if such an $r$ exists or $\bot$ otherwise
\end{algorithm}

Our guarantee about this algorithm is the following.
\begin{theorem}
\label{thm:trace}
For every $0 < \eps$ and  $\delta,\beta \in (0,1/n)$, $\PEstimateTrace_{\eps,\delta,\beta}(X)$ is $(\eps,\delta)$-DP.  Moreover, if $X_1, \ldots, X_n$ are i.i.d.\ from $\normal(0,\Sigma)$ for $\Sigma \succeq \id$, and
\[
n \geq \Omega \left( \frac{\log (\frac{1}{ \delta \beta}) }{\eps} \right)
\]
then with probability $1 - O(\beta)$, it outputs $T$ so that $\| \Sigma \|_2 \in [\xi T/d, \Xi d]$ for absolute constants $\xi,\Xi$.
\end{theorem}

\noindent To prove this theorem, we first show the following concentration bounds.  The first such bound says that if $\tr(\Sigma)$ is in the interval $I_{r^*}$, then almost all of the samples fall into one of the three buckets $I_{r^* - 1}, I_{r^*},$ and $I_{r^*+1}$.
\begin{lemma}
\label{lem:trace-lemma}
If $X_1, \ldots, X_n \sim \normal (0, \Sigma)$ for $n \geq \Omega(\log(1/\beta))$, and $r^*$ is such that $\tr(\Sigma) \in I_{r^*}$, then with probability $1-O(\beta)$, there are at most $2n/9$ samples such that $\| X_i \|_2^2 \not\in I_{r^* - 1} \cup I_{r^*} \cup I_{r^* + 1}$.
\end{lemma}
\begin{proof}
Let $r^*$ be so that $\tr (\Sigma) \in I_{r^*}$.
Clearly by assumption we have $r^* \in \{\log_C d, \log_C (\kappa d)\}$.
Recall that for each $X_i$, since $\| X_i \|_2^2 = Y_i^\top \Sigma Y_i$, where $Y_i \sim \normal (0, 1)$.
\begin{align*}
\pr{}{\| X_i \|_2^2 \not\in I_{r^* - 1} \cup I_{r^*} \cup I_{r^* + 1}} &\leq \pr{}{ Y_i^\top \Sigma Y_i < \tfrac{1}{C} \tr (\Sigma)} + \pr{}{ Y_i^\top \Sigma Y_i > C \tr (\Sigma)}
\end{align*}

By Corollary~\ref{cor:hanson-wright-extreme-lb} if $C$ is a sufficiently large constant, we have $\pr{}{Y_i^\top \Sigma Y_i < \frac{1}{C} \tr (\Sigma)} \leq \frac{1}{10}$, and by Theorem~\ref{thm:hanson-wright}, we know that, for a sufficiently large constant $C$,
\begin{align*}
\Pr \left[ Y_i^\top \Sigma Y > C \tr (\Sigma) \right] &\leq  \exp \left( - \min \left( \frac{(C - 1)^2 \tr (\Sigma)^2}{16 \| \Sigma \|_F^2}, \frac{(C - 1) \tr (\Sigma)}{4 \| \Sigma \|_2} \right) \right) \\
&\leq \exp \left( - \frac{C - 1}{4} \right) < \frac{1}{10} \; ,
\end{align*}
Hence we have 
\[
\Pr \left[ \| X_i \|_2^2 \not\in I_{r^* - 1} \cup I_{r^*} \cup I_{r^* + 1} \right] \leq \frac{1}{5} \; .
\]
Thus, by Hoeffding's bound, with probability $1 - O(\beta)$, the fraction of points $X_i$ which land in a bucket which is not $I_{r^* - 1}, I_{r^*},$ or $I_{r^* + 1}$ less than $1/5 + \sqrt{\log (1 / \beta) / 2n}$.  By our choice of $n$, this fraction is at most $2/9$.
\end{proof}

\noindent We can now return to the proof of Theorem~\ref{thm:trace}.
\begin{proof}[Proof of Theorem~\ref{thm:trace}]
Condition on the event that the conclusions of Theorem~\ref{thm:private-histogram} and Lemma~\ref{lem:trace-lemma} hold simultaneously, which, by our choice of $n$, occurs with probability at least $1 - O(\beta)$.
Let $r^*$ be so that $\tr (\Sigma) \in I_{r^*}$.
By definition we have $r^* \in \{\log_C(d) - 1, \ldots, \log_C (\kappa d) + 1 \}$.
Then, by Lemma~\ref{lem:trace-lemma}, we know that at least one of $I_{r^* - 1}, I_{r^*},$ or $I_{r^* + 1}$ must have probability mass at least $7 / 27$.
By Theorem~\ref{thm:private-histogram}, this implies that for this $r$, we know that $\cM(L)$ outputs at least 
$$\frac{7}{ 27} - O\left( \frac{\log(\frac{n}{\delta\beta}))}{\eps n} \right) > \frac{1}{4}$$ where the inequality follows from our choice of $n$.
Hence our algorithm will not terminate and output $\bot$.
Moreover, by Lemma~\ref{lem:trace-lemma}, no bucket other than $I_{r^* - 1}, I_{r^*},$ or $I_{r^* + 1}$ can have probability mass more than $2 / 9$, and hence by Theorem~\ref{thm:private-histogram}, $H(B)$ will not output a value $\geq 1/4$ for any other bucket.
Hence the output of $\PEstimateTrace$ is $C^r$ for $r \in \{ r^* - 1, r^*, r^* + 1\}$.
The last guarantee now follows simply because $\| \Sigma \|_2 \leq \tr (\Sigma) \leq d \| \Sigma \|_2$.
\end{proof}

\noindent
For the remainder of the section, we will simply say that \PEstimateTrace~\emph{succeeds} to mean that it satisfies the conclusion of Theorem~\ref{thm:trace}.

\subsection{Preconditioning with an Estimate of the Spectral Norm}

The algorithm in the previous section can be used to obtain a small interval $I = [a,b]$ that contains $\| \Sigma \|_2$.  Once we have such an interval, we can essentially run our recursive preconditioning subroutine $\PPreCond$ a few  times to essentially eliminate the largest eigenvector of $\Sigma$.

\begin{algorithm}[h!] \label{alg:precondition-interval}
\caption{Preconditioning given a range $\PPreCondRange_{\rho,\beta,I}(X_1, \ldots, X_n)$}
\KwIn{A set of $n$ samples $X_1, \ldots, X_n$ from an unknown Gaussian $\normal (0, \Sigma)$.  Parameters $\rho, \beta, > 0$, and an interval $I = [a, b] \subseteq [1, \infty)$ so that $a > 40 d^3$}
\KwOut{A private preconditioner for $\Sigma$, assuming that $\| \Sigma \|_2 \in I$} \vspace{10pt}

Set parameters $\beta' \gets \frac{\beta}{\log(b/a)}, \rho' \gets \frac{\rho}{\log(b/a)}$\\
Let $\kappa \gets b$ and let $V$ be the empty subspace\\
\While{$V$ is empty and $\kappa > a / 2$}{
	Let $(V,A) \gets \PPreCond_{\rho',\beta',\kappa, \frac{\kappa}{d^2}}(X_1, \ldots, X_n)$\\
	Let $\kappa \gets \frac{99}{100} \kappa$
}\vspace{10pt}

\Return the pair $(V,A)$, or $\bot$ if $V$ is empty.  Note that, by construction of $\PPreCond$
\[
A = \frac{2 d}{\kappa^{1/2}} \Pi_V + \Pi_{V^\perp}
\]
\end{algorithm}

\begin{lemma}
\label{lem:cov-precondition-interval}
For every $0 < \eps$ and $\delta,\beta \in (0,\frac1n)$ and , if $I = [a,b] \subseteq [40d^3,\infty)$ is an interval such that $\| \Sigma \|_2 \in [a,b]$, and $X_1,\dots,X_n$ are i.i.d.\ samples from $\normal(0,\Sigma)$ satisfying \eqref{eq:cov-cond1}--\eqref{eq:cov-cond3} for
$$
n \geq O\left( \frac{d^{3/2} \polylog(\frac{1}{\beta \rho}\cdot\frac{b}{a})}{\rho^{1/2}} \right)
$$
then $\PPreCondRange$ is $\rho$-zCDP and outputs $A$ such that
\begin{equation}
\label{eq:precondition-interval}
\left(1 - O\left( \frac{1}{d^2} \right) \right) \id \preceq A \Sigma A \preceq \left(5 d^2 + \frac{3 \| \Sigma \|_2}{4} \right) \id
\end{equation}
and a subspace $V$ such that if $U$ is a subspace satisfying $\| \Pi_{U} \Sigma \Pi_{U} \|_2 \leq 10d^2$ then $U' = U + V$ satisfies $\dim (U') > \dim (U)$ and $\| \Pi_U A \Sigma A \Pi_U \|_2 < 10 d^2$.
\end{lemma}
\begin{proof}
Privacy and~\eqref{eq:precondition-interval} follow immediately from Theorem~\ref{thm:dp-precondition} and composition of zCDP.
By Theorem~\ref{thm:dp-precondition} and the assumption $\| \Sigma \|_2 \in [a, b]$ $V$ is non-empty with probability at least $1 - O(\beta)$.
Let $U$ be as in the lemma statement, we will prove that the desired conditions hold.  Let $v \in V$ be a unit vector.  We have $v \not\in U$ as otherwise $v^\top \Pi_U \Sigma \Pi_U v^\top = v^\top \Sigma v \geq a / 4 > 10 d^3 \geq 10 d^2$, which contradicts the assumption that $\| \Pi_U \Sigma \Pi_U \|_2 < 10 d^2$.
Hence $U + V$ must have dimension strictly larger than $U$, which proves the claim.
\end{proof}

\subsection{Recursive Preconditioning}

We now show how to combine $\PEstimateTrace$ and $\PPreCondRange$ to recursively to obtain a private preconditioning matrix $A$ such that $\id \preceq \Sigma \preceq O(d^4) \id$.  Once we have reduced the condition number down to a polynomial, we can apply $\PGCEKappa$ to estimate the covariance.


\begin{algorithm}[h!] \label{alg:pce}
\caption{Privately estimating covariance $\PPCRange_{\eps,\delta,\beta}(X_1, \ldots, X_n)$}
\KwIn{A set of $n$ samples from an unknown Gaussian $\normal (0, \Sigma)$.  Parameters $\eps,\delta,\beta > 0$}
\KwOut{A symmetric matrix $A$ such that $\id \preceq A \Sigma A \preceq O(d^4) \id$.} \vspace{10pt}

Set parameters $\eps' = \eps/\sqrt{d \log(1/\delta)}~~~\delta' = \delta/d~~~\rho' = (\eps')^2 / \log(1/\delta)~~~\beta' = \beta / d$\\
Let $\xi, \Xi$ be the absolute constants from Theorem~\ref{thm:trace}\\
Let $X^{(1)}_i \gets X_i$ for $i = 1, \ldots, n$\\
\For{$j = 1, \ldots, d$}{
	Let $T^{(j)} \gets \PEstimateTrace_{\eps', \delta',\beta'}(X^{(j)}_1, \ldots, X^{(j)}_n)$\\
	Let $I^{(j)} \gets [a^{(j)}, b^{(j)}]$ for $a^{(j)} = \xi T^{(j)}$ and $b^{(j)} = \Xi d T^{(j)}$\\
	\If{$a^{(j)} < 40 d^3$}{\textbf{Break}}
	Let $A^{(j)} \gets \PPreCondRange_{\rho', \beta', I^{(j)}}(X^{(j)}_1, \ldots, X^{(j)}_n)$\\
	Let $X^{(j+1)}_i = A^{(j)} X^{(j)}_i$ for all $i = 1, \ldots, n$.
}

Let $J$ be the number of iterations the above loop ran for, and let
\[
A \gets 2 \prod_{j = 1}^J A^{(j)}
\]
%

\Return $A$
\end{algorithm}

\begin{theorem} \label{thm:ppc-nokappa}
For every $0 < \eps, \delta, \beta \leq 1$, $\PPCRange_{\eps,\delta,\beta}$ is $(O(\eps),O(\delta))$-DP and, if $X_1,\dots,X_n$ are i.i.d.\ samples from $\normal(0,\Sigma)$ and
$$
n \geq O\left( \frac{d^{3/2} \sqrt{\log(1/\delta)} \cdot \polylog(\frac{d \log(1/\delta)}{\beta \eps})}{\eps} \right)
$$
then with probability at least $1 - O(\beta)$ it outputs $A$ such that $\id \preceq A\Sigma A \preceq O(d^4) \id$.

\end{theorem}

\begin{proof}
We first prove that $\PGCE$ satisfies $(O(\eps), O(\delta))$-DP.  Using the fact that zCDP implies DP (Lemma~\ref{lem:dpdefns}) and composition of DP (Lemma~\ref{lem:dpcomp}), each iteration of the loop satisfies $(O(\eps/\sqrt{d \log(1/\delta)}, O(\delta))$-DP.  Therefore the whole algorithm satisfies $(O(\eps),O(\delta))$-DP by the strong composition theorem for DP (Lemma~\ref{lem:dpcomp}).

Condition on the event that $\PEstimateTrace$ and $\PPreCondRange$ succeed in each of the $d$ iterations, which, by our choice of parameters, occurs with probability at least $1 - O(\beta)$.  Under this assumption, we will show that $\id \preceq A \Sigma A \preceq O(d^4) \id$.  For the lower bound, by Lemma~\ref{lem:cov-precondition-interval}, since each instance of $\PEstimateTrace$ succeeds, we know that
\[
A \Sigma A \succeq 4 \left(1 - O \left( \frac{1}{d^2} \right) \right)^d \id \succeq 4 \left(1 - O \left( \frac{1}{d} \right) \right) \id \succeq \id
\]

We now prove the upper bound on $\| A \Sigma A \|_2$.  For each iteration $j = 1, \ldots, T$, let $V^{(j)}$ be the subspace found by $\PPreCondRange$.  Let $U^{(1)}$ be the empty subspace and let $U^{(j1)} = U^{(j-1)} + V^{(j-1)}$
By Lemma~\ref{lem:cov-precondition-interval}, we know that $\dim (U^{(j)}) > \dim(U^{(j - 1)})$ for all $j = 1, \ldots, T$.
Thus, in the final iteration, either $a < 40 d^3$, in which case $\| A \Sigma A \|_2 \leq 40 \Xi d^4 = O(d^4)$, or else $J = d$.  However, if this is the case then $U^{(J)} = \R^d$, and by Lemma~\ref{lem:cov-precondition-interval} we have that $\| A \Sigma A\|_2 \leq 10 d^2 = O(d^4)$.  This completes the proof.
\end{proof}

\subsection{Putting it All Together}

Finally, we combine the ingredients we've developed in this section to obtain a private algorithm for covariance estimation of arbitrarily ill-conditioned Gaussians.  First we present our algorithm, which uses $\PPCRange$ to obtain a preconditioning matrix $A$ such that $\id \preceq A \Sigma A \preceq O(d^4) \id$, and then we use $\PGCEKappa$ on the distribution $\normal(0, A\Sigma A)$ to obtain a final estimate of the covariance.  Since the distribution $\normal(0,A \Sigma A)$ is reasonably well conditioned, the sample complexity of $\PGCEKappa$ will be small.  The pseudocode is presented as Algorithm~\ref{alg:pgce-nokappa}.

\begin{algorithm}[h!] \label{alg:pgce-nokappa}
\caption{Private Covariance Estimator $\PGCE_{\eps,\delta, \beta}(X)$}
\KwIn{Samples $X_1,\dots,X_{n}$ from an unknown Gaussian $\normal(0, \Sigma)$.  Parameters $\eps,\delta,\beta> 0$.}
\KwOut{A matrix $\wh{\Sigma}$ such that $\| \Sigma - \wh{\Sigma} \|_\Sigma \leq \alpha$.}
\vspace{10pt}
Let $A \gets \PPCRange_{\eps,\delta,\beta,\kappa}(X_1,\dots,X_n)$ be the private preconditioner \\
Let $\kappa^* = O(d^4)$ be the upper bound on $\|A\Sigma A\|_2$ from Theorem~\ref{thm:ppc-nokappa}\\\vspace{10pt}
Let $Y_i \gets A X_i$ for $i = 1,\dots,n$\\
Let $\rho \gets \eps^2 / 8 \log(1/\delta)$\\
Let $\wt\Sigma \gets \PGCEKappa_{\rho, \beta, \kappa^*}(Y_1,\dots,Y_n)$\vspace{10pt}

\Return $\wh\Sigma = A^{-1} \wt\Sigma A^{-1}$
\end{algorithm}

\begin{theorem}
\label{thm:pce-nokappa}
For every $0 \leq \eps,\delta,\beta \leq 1$, $\PGCE_{\eps,\delta,\beta}(X)$ is $(O(\eps),O(\delta))$-DP and, when given
$$
n = O\left( \frac{d^{2} + \log(\frac{1}{\beta})}{\alpha^2} + \frac{d^2 \sqrt{\log(\frac{1}{\delta})} \cdot \polylog(\frac{d \log(1/\delta)}{\alpha \beta \eps})}{\alpha \eps}\right)
$$
samples $X_1,\dots,X_n \sim \normal(0,\Sigma)$ for $\Sigma \succeq \id$, with probability $1 - O(\beta)$, it outputs $\wh\Sigma$ with $\| \Sigma  - \wh\Sigma \|_{\Sigma} \leq O(\alpha)$.
\end{theorem}
\begin{proof}
Privacy follows from Theorem~\ref{thm:ppc-nokappa}, Theorem~\ref{thm:pce-kappa} and composition of differential privacy (Lemma~\ref{lem:dpcomp}).  Note that we use the fact that zCDP implies DP (Lemma~\ref{lem:dpdefns}).

By construction, the samples $Y_1,\dots,Y_n$ are i.i.d.\ from $\normal(0, A \Sigma A)$.  By Theorem~\ref{thm:ppc-nokappa}, and our choice of $n$, we have that except with probability $O(\beta)$, $A$ is such that $\id \preceq A\Sigma A^T \preceq \kappa^* \id$.  Therefore, combining the guarantees of \PGCEKappa~with our choice of $n$ we obtain that, except with probability $O(\beta)$, $\wt\Sigma$ satisfies $\| A \Sigma A - \wt\Sigma \|_{A \Sigma A} \leq O(\alpha)$.  The theorem now follows from the calculation $\| \Sigma - \wh \Sigma \|_{\Sigma} = \| \Sigma - A^{-1} \wt \Sigma A^{-1} \|_{\Sigma} = \| A \Sigma A - \wt \Sigma \|_{A \Sigma A} \leq O(\alpha).$
\end{proof}
\fi


\section{Private Mean Estimation for Gaussians} \label{sec:gaussian-mean}

Suppose we are given i.i.d.\  samples $X_1,\dots,X_n$, such that $X_i \sim \normal(\mu, \Sigma)$ where $\| \mu \|_2 \leq R$ is an unknown mean and $I \preceq \Sigma \preceq \kappa I$ is an unknown covariance matrix.  Our goal is to find an estimate $\wh\mu$ such that
$$
\| \mu - \wh\mu \|_\Sigma \leq O(\alpha)
$$
where $\| v \|_{\Sigma} = \| \Sigma^{-1/2} v \|_2$ is the Mahalanobis distance with respect to the covariance $\Sigma$.  This guarantee ensures
$$
\sd{\normal(\mu, \Sigma)}{\normal(\wh\mu, \Sigma)} = O(\alpha).
$$

When $\kappa$ is a constant, we can obtain such a guarantee in a relatively straightforward way by applying the mean-estimation procedure for univariate Gaussians due to Karwa and Vadhan~\cite{KarwaV18} to each coordinate.  To handle large values of $\kappa$, we combine their procedure with our procedure for privately learning a strong approximation to the covariance matrix.

\subsection{Mean Estimation for Well-Conditioned Gaussians}

\newcommand{\kvmean}{\ensuremath{\textsc{KVMean}}}
\newcommand{\NaivePME}{\ensuremath{\textsc{NaivePME}}}
\newcommand{\PME}{\ensuremath{\textsc{PME}}}
\newcommand{\PCE}{\ensuremath{\textsc{PCE}}}

We start with the following algorithm for learning the mean of a univariate Gaussian, which is a trivial variant of the algorithm of Karwa and Vadhan~\cite{KarwaV18} to the definition of zCDP.

\begin{thm}[Variant \cite{KarwaV18}] \label{thm:kvmean}
For every $\eps, \delta, \alpha, \beta, R, \sigma > 0$, there is an $\frac{\eps^2}{2}$-zCDP algorithm $\kvmean_{\eps, \alpha, \beta, R, \kappa}(X)$ and an
$$
n = O\left( \frac{\log(\frac{1}{\beta})}{\alpha^2} + \frac{\log(\frac{\log R}{\alpha \beta \eps})}{\alpha \eps} + \frac{\log^{1/2}(\frac{R}{\beta})}{\eps} \right)
$$
such that if $X = (X_1,\dots,X_n)$ are i.i.d.\  samples from $\normal(\mu, \sigma^2)$ for $|\mu| \leq R$ and $1 \leq \sigma^2 \leq \kappa$ then, with probability at least $1-\beta$, $\kvmean$ outputs $\wh{\mu}$ such that $| \mu - \wh\mu | \leq \alpha \kappa$.  
\ifnum\epsdeltastuff=1 There is also an analogous $(\eps,\delta)$-DP algorithm with sample complexity
$$
n = O\left( \frac{\log(\frac{1}{\beta})}{\alpha^2} + \frac{\log(\frac{\log (1/\delta)}{\alpha \beta \eps})}{\alpha \eps} + \frac{\log(\frac{1}{\beta\delta})}{\eps} \right).
$$\fi
\end{thm}

Note that the algorithm only needs an \emph{upper bound} $\kappa$ on the true variance $\sigma^2$ as a parameter.  However, since the error guarantees depend on this upper bound, the upper bound needs to be reasonably tight in order to get a useful estimate of the mean.

We will describe our na\"ive algorithm for the case of $\rho$-zCDP since the parameters are cleaner.  We could also obtain an $(\eps,\delta)$-DP version using the $(\eps,\delta)$-DP version of \NaivePME~and setting parameters appropriately.

\begin{algorithm}[h!] \label{alg:naivemean}
\caption{Na\"ive Private Mean Estimator $\NaivePME_{\rho, \alpha, \beta, R, \kappa}(X)$}
\KwIn{Samples $X_1,\dots,X_n \in \R^d$ from a $d$-variate Gaussian.  Parameters $\rho, \alpha, \beta, R, \kappa > 0$.}
\KwOut{A vector $\wh{\mu}$ such that $\| \mu - \wh{\mu} \|_\Sigma \leq \alpha$.}

Let $\rho' \gets \rho/d~~~~\alpha' \gets \alpha / \kappa\sqrt{d}~~~~\beta' \gets \beta/d$ \\
\For{$j = 1,\dots,d$}{
		Let $\wh{\mu}_{j} \gets \kvmean_{\rho', \alpha', \beta', R, \kappa}(X_{1,j},\dots,X_{n,j})$
}

\Return $\wh\mu = (\wh\mu_1,\dots,\wh\mu_d)$
\end{algorithm}

\begin{thm}
For every $\rho, \alpha, \beta, R, \kappa > 0$, the algorithm $\NaivePME_{\rho, \alpha, \beta, R, \kappa}$ is $\rho$-zCDP and there is an
$$
n = O\left( \frac{\kappa^2 d \log(\frac{d}{\beta})}{\alpha^2} + \frac{\kappa d \log(\frac{\kappa d \log R}{\alpha \beta \rho})}{\alpha \rho^{1/2}} + \frac{\sqrt{d} \log^{1/2}(\frac{Rd}{\beta})}{\rho^{1/2}} \right)
$$
such that if $X_1,\dots,X_n \sim \normal(\mu, \Sigma)$ for $\| \mu \|_2 \leq R$ and $I \preceq \Sigma \preceq \kappa I$ then, with probability at least $1-\beta$, $\NaivePME$ outputs $\wh\mu$ such that $\| \mu - \wh\mu \|_{\Sigma} \leq \alpha.$  \ifnum\epsdeltastuff=1 There is also an analogous $(\eps,\delta)$-DP algorithm with sample complexity
$$
n = O\left( \frac{\kappa^2 d \log(\frac{d}{\beta})}{\alpha^2} + \frac{\kappa d \log(\frac{\kappa d \log (d/\delta)}{\alpha \beta \eps})}{\alpha \eps} + \frac{\sqrt{d}\log^{3/2}(\frac{1}{\beta \delta})}{\eps} \right).
$$\fi
\end{thm}

\begin{proof}
The fact that the algorithm satisfies $\rho$-zCDP follows immediately from the assumed privacy of $\kvmean$ and the composition property for zCDP (Lemma~\ref{lem:dpcomp}).

Next we argue that with probability at least $1-\beta$, for every coordinate $j = 1,\dots,d$, we have $| \mu_j - \wh\mu_j | \leq \alpha/\sqrt{d}$.  Observe that, since $X_1,\dots,X_n$ are distributed as $\normal(\mu, \Sigma)$, the $j$-th coordinates $X_{1,j},\dots,X_{n,j}$ are distributed as $\normal(\mu_j, \Sigma_{jj})$ and, by assumption $|\mu_j| \leq R$ and $1 \leq \Sigma_{jj} \leq \kappa$.  Thus, by Theorem~\ref{thm:kvmean}, we have $| \mu_j - \wh\mu_j | \leq \alpha' \kappa = \alpha/\sqrt{d}$ except with probability at most $\beta' = \beta/d$.  The statement now follows by a union bound.

Assuming that every coordinate-wise estimate is correct up to $\alpha/\sqrt{d}$, we have 
\begin{align*}
\| \mu - \wh\mu \|_{\Sigma}
={} \|\Sigma^{-1/2} (\mu - \wh\mu)   \|_2 
\leq{} \| \Sigma^{-1/2} \|_2 \cdot \| \mu - \wh\mu \|_2 
\leq{} \alpha
\end{align*}
where the final equality uses the coordinate-wise bound on $\mu - \wh\mu$ and the fact that $I \preceq \Sigma$.
To complete the proof, we can plug our choices of $\rho', \alpha', \beta'$ into the sample complexity bound for \kvmean~from Theorem~\ref{thm:kvmean}.

The proof and algorithm for the final statement of the theorem regarding $(\eps,\delta)$-DP are completely analogous.  This completes the proof of the theorem.
\end{proof}

\subsection{An Algorithm for General Gaussians}

If $\Sigma$ were known, then we could easily perform mean estimation without dependence on $\kappa$ simply by applying $\Sigma^{-1/2}$ to each sample and running \NaivePME.  Specifically, if $X_i \sim \normal(\mu, \Sigma)$ then $\Sigma^{-1/2} X_i \sim \normal(\Sigma^{-1/2} \mu, I)$.  Then applying \NaivePME~we would obtain $\wh\mu$ such that $ \| \mu - \wh\mu \|_{\Sigma} = \| \Sigma^{-1/2} (\mu - \wh\mu)   \|_2 \leq \alpha$ and the sample complexity would be independent of $\kappa$.  Using the private preconditioner from the previous section, we can obtain a good enough approximation to $\Sigma^{-1/2}$ to carry out this reduction.

\begin{algorithm}[h!] \label{alg:privatemean}
\caption{Private Mean Estimator $\PME_{\rho, \alpha, \beta, R, \kappa}(X)$}
\KwIn{Samples $X_1,\dots,X_{3n} \in \R^d$ from a $d$-variate Gaussian $\normal(\mu, \Sigma)$ with unknown mean and covariance.  Parameters $\rho, \alpha, \beta, R, \kappa > 0$.}
\KwOut{A vector $\wh{\mu}$ such that $\| \mu - \wh{\mu} \|_\Sigma \leq \alpha$.}
\vspace{10pt}
For $i = 1,\dots,n$, let $Z_i = \frac{1}{\sqrt{2}}(X_{2i} - X_{2i-1})$ \\
Let $A \gets \PPC_{\rho, \beta, \kappa}(Z_1,\dots,Z_{n})$ \\\vspace{10pt}

For $i = 1,\dots,n$, let $Y_i = AX_{2n+i}$ \\
Let $\wt\mu \gets \NaivePME_{\rho, \alpha, \beta, 1000R, 1000}(Y_1,\dots,Y_n)$\\\vspace{10pt}

\Return $\wh\mu \gets A^{-1} \wt\mu$
\end{algorithm}

We capture the properties of $\PME$ in the following theorem
\begin{thm} \label{thm:meanfinal}
For every $\rho, \alpha, \beta, R, \kappa > 0$, the algorithm $\PME_{\rho, \alpha, \beta, R, \kappa}$ is $2\rho$-zCDP and there is an
$$
n = O\left( \frac{d \log(\frac{d}{\beta})}{\alpha^2} + \frac{d \log(\frac{d \log R}{\alpha \beta \rho})}{\alpha \rho^{1/2}} + \frac{\sqrt{d} \log^{1/2}(\frac{Rd}{\beta})}{\rho^{1/2}} + n_{\PPC} \right)
$$
such that if $X_1,\dots,X_n \sim \normal(\mu, \Sigma)$ for $\| \mu \|_2 \leq R$ and $I \preceq \Sigma \preceq \kappa I$ then, with probability at least $1-2\beta$, $\PME$ outputs $\wh\mu$ such that $\| \mu - \wh\mu \|_{\Sigma} \leq \alpha.$  In the above, $n_{\PPC}$ is the sample complexity required by $\PPC_{\rho,\beta,\kappa}$ (Theorem~\ref{thm:cov-kappa}).  \ifnum\epsdeltastuff=1 There is also an analogous $(\eps,\delta)$-DP algorithm with sample complexity
$$
n = O\left( \frac{d \log(\frac{d}{\beta})}{\alpha^2} + \frac{d \log(\frac{d \log (d/\delta)}{\alpha \beta \eps})}{\alpha \eps} + \frac{\sqrt{d}\log^{3/2}(\frac{1}{\beta \delta})}{\eps} + n_{\textsc{PPCU}} \right)
$$
where $n_{\textsc{PPCU}}$ is the sample complexity required by $\PPCRange_{\eps,\delta,\beta}$ (Theorem~\ref{thm:ppc-nokappa}).\fi
\end{thm}

\begin{proof}
Privacy will follow immediately from the composition property of $2\rho$-zCDP and the assumed privacy of \PPC~and \NaivePME.  The sample complexity bound will also follow immediately from the sample complexity bounds for \PCE~and \NaivePME.  Thus, we focus on proving that $\| \mu - \wh\mu \|_{\Sigma} \leq \alpha$.

Since $X_1,\dots,X_{2n}$ are i.i.d.\  from $\normal(\mu, \Sigma)$, the values $Z_1,\dots,Z_n$ are i.i.d.\  from $\normal(0, \Sigma)$.  Therefore, with probability at least $1-\beta$, $\PPC(Z_1,\dots,Z_n)$ returns a matrix $A$ such that $I \preceq A\Sigma A \preceq 1000 I$.  Note that since $I \preceq \Sigma$ and $A\Sigma A \preceq 1000$ we have $\|A\|_{2} \leq 1000$.  

Now, since the samples $X_{2n+1},\dots,X_{3n}$ are i.i.d.\  from $\normal(\mu, \Sigma)$, the values $Y_1,\dots,Y_n$ are i.i.d.\  from $\normal(A\mu, A\Sigma A)$.  Note that $\|A \mu\|_2 \leq \|A\|_2 \|\mu\|_2 \leq 1000R$ and, by assumption, $I \preceq A\Sigma A \preceq 1000 I$.  When we apply $\NaivePME_{\rho, \alpha, \beta, 1000 R, 1000}$ to $Y_1,\dots,Y_n$, with probability at least $1-\beta$ we will obtain $\wt\mu$ such that $\| A\mu - \wt\mu \|_{A \Sigma A} \leq \alpha$.  
Finally, we can write 
$
\| \mu - \wh\mu \|_{\Sigma} = \| A \mu  - A \wh\mu \|_{A \Sigma A} =  \| A \mu  - \wt\mu \|_{A \Sigma A} \leq \alpha.
$
The theorem now follows by a union bound over the two possible failure events.
\ifnum\epsdeltastuff=1 The proof and algorithm for the final statement of the theorem regarding $(\eps,\delta)$-DP are completely analogous.  This completes the proof of the theorem.\fi
\end{proof}

\section{Privately Learning Product Distributions}
\label{sec:product}
In this section we introduce and analyze our algorithm for learning a product distribution $P$ over $\zo^d$ in total variation distance, thereby proving Theorem~\ref{thm:mainproduct} in the introduction. The pseudocode appears in Algorithm~\ref{alg:ppde}. For simplicity of presentation, we assume that the product distribution has mean that is bounded coordinate-wise by $\frac12$ (i.e.~$\ex{}{P} \preceq \frac12$), although we emphasize that this assumption is essentially without loss of generality, and can easily be removed while paying only a constant factor in the sample complexity.

\subsection{A Private Product-Distribution Estimator}

To describe the algorithm, we need to introduce notation for the \emph{truncated mean}.  Given a dataset element (a vector) $X_i \in \zo^{d}$ and $B \geq 0$, we use
\begin{equation*}
\trunc_{B}(X_i) = 
\begin{cases}
X_i &\textrm{if $\|X_i\|_2 \leq B$}\\
\frac{B}{\|X_i\|_{2}} \cdot  X_i &\textrm{if $\|X_i\|_{2} > B$}
\end{cases}
\end{equation*}
to denote the truncation of $x$ to an $\ell_2$-ball of radius $B$.  Given a dataset $X = (X_1,\dots,X_{m}) \in \zo^{m \times d}$ and $B > 0$, we use
\begin{equation*}
\tmean_{B}(X) = \frac{1}{m} \sum_{i=1}^{m} \trunc_{B}(X_i)
\end{equation*}
to denote the mean of the truncated vectors.
Observe that the $\ell_2$-sensitivity of
$\tmean_{B}$ is $\frac{B}{m}$, while the $\ell_2$-sensitivity of the untruncated mean is infinite.  Note that
$\tmean_{B}(X) = \frac{1}{m} \sum_{i=1}^{m} X_{i}$
unless $\|X_{i}\|_2 > B$ for some $i$.  If one
of the inputs to $\tmean_{B}$ does not satisfy
the norm bound then we will say,
``truncation occurred,'' as a shorthand.

We also use the following notational conventions:
Given a dataset element $X_i \in \zo^d$, we will use the array notation $X_i[j]$ to refer to its $j$-th coordinate, and the notation $X_i[S] = (X_i[j])_{j \in S}$ to refer to the vector $X$ restricted to the subset of coordinates $S \subseteq [d]$.  Given a dataset $X = (X_1,\dots,X_m)$, we use the notation $X[S] = (X_1[S],\dots,X_m[S])$ to refer to the dataset consisting of each $X_i$ restricted to the subset of coordinates $S \subseteq [d]$.

\begin{algorithm}[h!] \label{alg:ppde}
\caption{Private Product-Distribution Estimator $\PPDE_{\rho, \alpha, \beta}(X)$}
\KwIn{Samples $X_1,\dots,X_n \in \zo^{d}$ from an unknown product distribution $P$ satisfying $\ex{}{P} \preceq \frac12$.  Parameters $\rho, \alpha, \beta > 0$.}
\KwOut{A product distribution $Q$ over $\{0,1\}^d$ such that $\SD(P,Q) \leq \alpha$.} \vspace{10pt}
Set parameters:
$
R \gets \log_{2}\left(d/2\right)~~~
c \gets 128 \log^{5/4}(d / \alpha\beta(2\rho)^{1/2})~~~c' \gets 128 \log^{3}(dR/\beta)~~~ 
m \gets \frac{c'd}{\alpha^2} + \frac{cd}{\alpha (2\rho)^{1/2}}
$\\

Split $X$ into $R+1$ blocks of $m$ samples each, denoted $X^{r} = (X^{r}_{1},\dots,X^{r}_{m})$ \\
(Halt and output $\bot$ if $n$ is too small.)

Let $q[j] \gets 0$ for every $j \in [d]$, and let $S_1 = [d]$, $u_1 \gets \frac{1}{2}$, $\tau_1 \gets \frac{3}{16}$, and $r \gets 1$ \\\vspace{10pt}
\tcp{Partitioning Rounds}
    \While{$u_r |S_r| \geq 1$}{
	Let $S_{r+1} \gets \emptyset$ \\
	Let $B_r \gets \sqrt{6 u_r |S_r| \log(mR/\beta)}$ \\
    Let $q_r[S_r] \getsr \tmean_{B_r}(X^r[S_r]) + \normal\left(0, \frac{B_r^2}{ 2 \rho m^2} \cdot \id \right)$
        \For{$j \in S_r$}{
        	\eIf{$q_r[j] < \tau_r$}{
            	Add $j$ to $S_{r+1}$
            }
            {
            	Set $q[j] \gets q_r[j]$
            }
        }
    Let $u_{r+1} \gets \frac12 u_{r}$, $\tau_{r+1} \gets \frac12 \tau_{r}$, and $r \gets r+1$
    }\vspace{10pt}
\tcp{Final Round}
    \If{$|S_r| \geq 1$}{
    	Let $B_r \gets \sqrt{6\log(m/\beta)}$\\
    Let $q[S_r] \getsr \tmean_{B_r}(X^r[S_r]) + \normal\left(0, \frac{B_r^2}{ 2 \rho m^2} \cdot \id \right)$
    }\vspace{10pt}
    \Return $Q = \Ber(q[1]) \otimes \dots \otimes \Ber(q[d])$
\end{algorithm}

The privacy analysis of Algorithm~\ref{alg:ppde} is straightforward, based on privacy of the Gaussian mechanism and bounded sensitivity of the truncated mean.

\begin{thm} \label{thm:ppde_dp}
For every $\rho,\alpha,\beta > 0$, $\PPDE_{\rho, \alpha, \beta}(X)$ satisfies $\rho$-zCDP.
\end{thm}
\begin{proof}
Since each individual's data is used only to compute $\tmean_{B_{r}}(X^r)$ for a single round $r$, privacy follows immediately from Lemma~\ref{lem:gaussiandp} and from observing that the $\ell_2$-sensitivity of $\tmean_{B}$ is $\frac{B}{n}$.  Note that, since a disjoint set of samples $X^r$ is used for each round $r$, each sample only affects a single one of the rounds, so we do not need to apply composition.\end{proof}

\subsection{Accuracy Analysis for PPDE}

In this section we prove the following theorem bounding the sample complexity required by \PPDE{} to learn a product distribution up to $\alpha$ in total variation distance.
\begin{thm} \label{thm:ppde_acc}
    For every $d \in \N$, every product distribution
    $P$ over $\zo^{d}$, and every $\rho,\alpha, \beta > 0$,
        if $X = (X_1,\dots,X_n)$ are independent
        samples from $P$ for
    $$
        n = \tilde{O}\left(\frac{d}{\alpha^2} +
        \frac{d}{\alpha \sqrt{\rho}} \right),
    $$
    then with probability at least $1-O(\beta)$,
    $\PPDE_{\rho, \alpha, \beta}(X)$ outputs $Q$,
    such that $\SD(P,Q) \leq O(\alpha)$.  The notation $\tilde{O}(\cdot)$ hides polylogarithmic factors in $d, \frac{1}{\alpha}, \frac{1}{\beta},$ and $\frac{1}{\rho}$.
\end{thm}
 
Before proving the theorem, we will introduce or recall a few useful tools and inequalities.

\paragraph{Distances Between Distributions.}
We use several notions of distance between distributions.  
\begin{defn}
If $P,Q$ are distributions, then
\begin{itemize}
\item the \emph{statistical distance} is $\SD(P,Q) = \frac12 \sum_{x} | P(x) - Q(x) |$, 
\item the \emph{$\chi^2$-divergence} is $\CS(P \| Q) = \sum_{x} \frac{(P(x) - Q(x))^2}{Q(x)}$, and
\item the \emph{KL-divergence} is $\KL(P \| Q) = \sum_{x} P(x) \log \frac{P(x)}{Q(x)}$.
\end{itemize}
\end{defn}

For product distributions $P = P_1 \otimes \dots \otimes P_k$ and $Q = Q_1 \otimes \dots \otimes Q_k$, the $\chi^2$ and $\KL$ divergences are additive, and the statistical distance is subadditive.  Specifically,
\begin{lem} \label{lem:sd-ub}
	Let $P = P_1 \otimes \dots \otimes P_k$ and
    $Q = Q_1 \otimes \dots \otimes Q_k$ be two
    product distributions. Then
    \begin{itemize}
    \item $\sd{P}{Q} \leq \sum_{j=1}^{d} \SD(P_j,Q_j),$
    \item $\CS(P \| Q) \leq \sum_{j=1}^{d} \CS(P_j \| Q_j)$, and
    \item $\KL(P \| Q) = \sum_{j=1}^{d} \KL(P_j \| Q_j).$
    \end{itemize}
\end{lem}

\noindent The three definitions also satisfy some useful relationships.
\begin{lem} \label{lem:pinsker}
For any two distributions $P,Q$ we have $2 \cdot \SD(P,Q)^2 \leq \KL(P \| Q) \leq \CS(P \| Q)$.
\end{lem}

\paragraph{Tail Bounds.}

We need a couple of useful tail bounds for sums of independent Bernoulli random variables.  The first lemma is a useful form of the Chernoff bound.

\begin{lem}\label{thm:chernoff}
For every $p > 0$, if $X_1,\dots,X_m$ are i.i.d.\ samples from $\Ber(p)$
then for every $\eps > 0$
\begin{equation*}
    \pr{}{\frac{1}{m}\sum\limits_{i=1}^{m}{X_i} \geq
    	p + \eps} \leq e^{-\KL(p+\eps||p)\cdot m}
    ~~~\textrm{and}~~~\pr{}{\frac{1}{m}\sum\limits_{i=1}^{m}{X_i} \leq
        p - \eps} \leq e^{-\KL(p-\eps||p)\cdot m}
\end{equation*}
\end{lem}

The next lemma follows easily from a Chernoff bound.

\begin{lem}\label{fact:chernoff}
	Suppose $X_1,\dots,X_k$ are
    sampled i.i.d.\ from a product
    distribution $P$ over $\{0,1\}^t$, where the mean of each coordinate is upper
    bounded by $p$ (i.e. $\ex{}{P}\preceq p$).  Then
    \begin{enumerate}
    \item if $pt > 1$, then for each $i$,
        $\pr{}{\llnorm{X_i}^2 \geq
    	pt\left(1+3\log(\frac{k}{\beta})\right)}
        \leq \frac{\beta}{k}$, and
    \item if $pt \leq 1$, then for each $i$,
        $\pr{}{\llnorm{X_i}^2 \geq
    	6\log(\frac{k}{\beta})}
        \leq \frac{\beta}{k}.$
    \end{enumerate}
\end{lem}

\subsubsection{Analysis of the Partitioning Rounds}
In this section we analyze the progress made during the partitioning rounds.  We show two properties: (1) any coordinate $j$ such that $q[j]$ was set during the partitioning rounds has small error and (2) any coordinate $j$ such that $q[j]$ was not set in the partitioning rounds has a small mean.  We capture the properties of the partitioning rounds that will be necessary for the proof of Theorem~\ref{thm:ppde_acc} in the following lemma.
\begin{lem} \label{lem:ppde_attenuation}
    If $X^{1},\dots,X^{R}$ each contain at least
    $$
    m =\frac{128 d \log^3 (dR/\beta )}
        {\alpha^2}+ \frac{128 d \log^{5/4}
        (d/\alpha\beta(2\rho)^{1/2})}{\alpha \rho^{1/2}}
    $$
    i.i.d. samples from $P$, then, with probability
    at least $1-O(\beta)$, in every partitioning round
    $r = 1,\dots,R$:
    \begin{enumerate}
        \item If a coordinate $j$ does not go to
            the next round (i.e.~$j \in S_{r}$ but
            $j \not\in S_{r+1}$) then $q[j]$ has
            small $\chi^2$-divergence with $p[j]$, 
            $$
            \frac{4(p[j] - q[j])^2}{q[j]} \leq \frac{\alpha^2}{d}.
            $$
            Thus, if $S_{A}$ consists of all coordinates
            such that $q[j]$ is set in one of the
            partitioning rounds,
            $\SD(P[S_{A}], Q[S_{A}]) \leq \alpha.$
    
        \item If a coordinate $j$ does go to the next
            round (i.e.~$j \in S_{r}, S_{r+1}$), then
            $p[j]$ is small,
            $$
            p[j] \leq u_{r+1} = \frac{u_{r}}{2}.
            $$
    \end{enumerate}
\end{lem}
\begin{proof}
We will prove the lemma by induction on $r$
(taking a union bound over the events that
one of the two conditions fails in a given
round $r$).  Therefore, we will assume that
in every round $r$, $p[j] \leq u_{r}$ for
every $j \in S_{r}$ and prove that if this
bound holds then the two conditions in the
lemma hold.  For the base of the induction,
observe that, by assumption, $p[j] \leq u_{1} = \frac12$
for every $j \in S_{1} = [d]$.  In what
follows we fix an arbitrary round $r \in [R]$.
Throughout the proof, we will use the notation
$\tilde{p}_{r} = \frac{1}{m} \sum_{i = 1}^{m} X^{r}_{i}$
to denote the empirical mean of the $r$-th
block of samples.

\begin{clm} \label{clm:atten-samp-err}
If $\tilde{p}_r[j] = \frac{1}{m} \sum_{i=1}^{m} X^{r}_{i}[j]$
and $p_j > \frac{1}{d}$, then
with probability at least $1-\frac{2\beta}{R}$,
$$
\forall j \in S_{r}~~\left| p[j] - \tilde{p}_r[j] \right|
    \leq \sqrt{\frac{4p[j]
    \log\left( \frac{dR}{\beta} \right)}{m}}
$$
\end{clm}
\begin{proof}[Proof of Claim \ref{clm:atten-samp-err}]
    We use a Chernoff Bound (Theorem \ref{thm:chernoff}), and
    facts that 
    $$\forall \gamma > 0~~~~\KL(p+\gamma||p) \geq \frac{\gamma^2}
    {2(p+\gamma)}~~~\textrm{and}~~~\KL(p-\gamma||p)
    \geq \frac{\gamma^2}{2p},$$ and set
    $$\gamma = \sqrt{\frac{4p[j] \log \left( \frac{dR }{ \beta} \right)}{m}}.$$
    Note that when $p[j] > \frac{1}{d}$, due
    to our choice of parameters, $\gamma \leq p[j]$.
    Therefore, $2(p[j]+\gamma) \leq 4p[j]$.
    Finally, taking a union bound over the cases
    when $\tilde{p}_r[j] \leq p[j] - \gamma$ and when
    $\tilde{p}_r[j] \geq p[j] + \gamma$, we
    prove the claim.
\end{proof}

\begin{clm} \label{clm:atten-no-trunc}
    With probability at least $1-\frac{\beta}{R}$,
    for every $X^{r}_{i} \in X^{r}$, $\|X^{r}_{i}\|_{2} \leq B_{r}$,
    so no rows of $X^{r}$ are truncated in the
    computation of $\tmean_{B_{r}}(X^{r})$.
\end{clm}
\begin{proof}[Proof of Claim \ref{clm:atten-no-trunc}]
	By assumption, all marginals specified
    by $S_r$ are upper bounded by $u_r$. Now,
    the expected value of $\|X^r_i\|_2^2$ is
    at most $u_r|S_r|$. Since
	$B_r = \sqrt{u_r |S_r| 6\log(mR/\beta)}$,
    we know that $B_r \geq \sqrt{u_r |S_r| \left(1+3\log(mR/\beta)\right)}$.
    The claim now follows from a Chernoff Bound (Lemma \ref{fact:chernoff})
    and a union bound over the entries of $X^r$.
\end{proof}
    
\begin{clm} \label{clm:atten-priv-err}
With probability at least $1-\frac{2\beta}{R}$,
$$
\forall j \in S_{r}~~\left| \tilde{p}_r[j] - q_{r}[j] \right| \leq
	\sqrt{\frac{6 u_r |S_r|
	\log\left(\frac{mR}{\beta}\right)
    \log\left(\frac{2dR}{\beta}\right)}{\rho m^2}}
$$
\end{clm}
\begin{proof}[Proof of Claim \ref{clm:atten-priv-err}]
	We assume that all marginals specified
    by $S_r$ are upper bounded by $u_r$. From
    Claim \ref{clm:atten-no-trunc}, we know that,
    with probability at least $1-\beta/R$, there is no truncation,
    so $\tmean_{B_{r}}(X^r[S_r]) = \frac{1}{m}
    \sum_{X_i^{r} \in X^{r}} X^{r}[S_r] = \tilde{p}_r[S_{r}]$.
    So, the Gaussian noise
    is added to $\tilde{p}_r[j]$ for each $j \in S_r$.
    Therefore, the only source of error here is the
    Gaussian noise. Using the standard tail bound
    for zero-mean Gaussians (Lemma~\ref{fact:gaussian-error}),
    with the following parameters,
    $$
    \sigma = \sqrt{\frac{3 u_r |S_r|
	\log\left(\frac{mR}{\beta}\right)}{\rho m^2}}
    ~~~\textrm{and}~~~ t = \sqrt{2\log\left(\frac{2dR}{\beta}\right)},$$
    and taking a union bound over all
    coordinates in $S_r$, and the event
    of truncation, we obtain the claim.
\end{proof}

Plugging our choice of $m$ into
Claims~\ref{clm:atten-samp-err} and~\ref{clm:atten-priv-err}, applying the triangle inequality, and simplifying, we get
that (with high probability),
$$
\abs{p[j] - q_r[j]} \leq \frac{\alpha}{\log^{1/4}
	\left(\frac{d}{\beta}\right)}\sqrt{\frac{u_r}{d}}. 
$$
To simplify our calculations, we will define
$$
e_r = \frac{\alpha}{\log^{1/4}
	\left(\frac{d}{\beta}\right)}\sqrt{\frac{u_r}{d}} 
$$
to denote the above bound on $|p[j] - q[j]|$ in round $r$.

\begin{clm}\label{clm:atten-chi-ub}
    For all $j \in S_r$, with probability,
    at least, $1-\frac{4\beta}{R}$,
    $$\chi^2\left(p[j],q_r[j]\right) \leq
        \frac{4(p[j]-q_r[j])^2}{q_r[j]}.$$
\end{clm}
\begin{proof}
    For every $r$, and $d$ more than some absolute constant,
    $\abs{e_r} \leq \frac{1}{4},$. Also, by assumption, $p[j] \leq \frac{1}{2},$
    for all $j \in [d]$. Therefore, for
    every $r$, and every $j \in S_r$,
    \begin{align*}
        \chi^2\left(p[j],q_r[j]\right) =&~
            \frac{(p[j]-q_r[j])^2}{q_r[j]} + \frac{(p[j]-q_r[j])^2}{1-q_r[j]}\\
            =&~ \frac{(p[j]-q_r[j])^2}{q_r[j](1-q_r[j])}\\
            \leq&~ \frac{4(p[j]-q_r[j])^2}{q_r[j]}
    \end{align*}
\end{proof}

\begin{clm} \label{clm:atten-marg-fix}
    With probability at least
    $1-\frac{4\beta}{R}$,
    for every $j \in S_{r}$,
    $$
    q_{r}[j] \geq \tau_{r} \Longrightarrow
        \frac{4(p[j] - q[j])^2}{q[j]} \leq \frac{\alpha^2}{d}.
    $$
\end{clm}
\begin{proof}[Proof of Claim \ref{clm:atten-marg-fix}]
	We want to show the following inequality.
    \begin{align*}
    	\frac{4\left(p_j-q_r[j]\right)^2}{q_r[j]} &\leq
    		\frac{\alpha^2}{d}
    \end{align*}
    We know
    that $\abs{p_j-q_r[j]} \leq e_r$ with
    high probability. Thus, we need to show that if
    $q_r[j] \geq \tau_r$, the following
    inequality holds:
    \begin{align*}
    	\frac{4e_r^2}{q_r[j]} \leq
    		\frac{\alpha^2}{d}
        \iff \frac{4de_r^2}{\alpha^2} &\leq q_r[j].
    \end{align*}
    We now show that the left-hand side is at most $\tau_r$, which completes the proof.
    In the algorithm, we have $\tau_r = \frac{3}{4}u_{r+1}$:
    \begin{align*}
    	\frac{4de_r^2}{\alpha^2} \leq \frac{3}{4}u_{r+1} 
        \iff \frac{4u_r}{\log^{1/2}
			\left(\frac{d}{\beta}\right)}
        	\leq \frac{3}{4}u_{r+1} 
        \iff \frac{16}{3\log^{1/2}
			\left(\frac{d}{\beta}\right)} \leq \frac{1}{2}.
    \end{align*}
    Note that the final inequality is satisfied
    as long as $d$ is larger than some absolute constant.
\end{proof}

\begin{clm} \label{clm:atten-next-rnd}
    With probability at least $1-\frac{4\beta}{R}$,
    for every $j \in S_{r}$,
    $$
    q_{r}[j] < \tau_{r} \Longrightarrow p[j] \leq u_{r+1} = \frac{u_{r}}{2}.
    $$
\end{clm}
\begin{proof}[Proof of Claim \ref{clm:atten-next-rnd}]
    We know that with high probability,
    $p_j \leq q_r[j] + e_r$. But since
    $q_r[j] < \tau_i$, we know that
    $p_j < \tau_r + e_r$. Also,
    $\tau_r = \frac{3}{4}u_{r+1}$. Then it is sufficient
    to show the following.
    \begin{align}
    	&e_r \leq \frac{u_{r+1}}{4}\nonumber\\
        \iff{} &\frac{\alpha}{\log^{1/4}
			\left(\frac{d}{\beta}\right)}
        	\sqrt{\frac{u_r}{d}} \leq \frac{u_{r+1}}{4}\nonumber\\
        \iff{} &\frac{16\alpha^2}{d u_1 \log^{1/2}
			\left(\frac{d}{\beta}\right)}
        	\leq \left(\frac12\right)^{r+1}\nonumber\\
        \iff{} &\left(\frac{16\alpha^2}{d u_1 \log^{1/2}
			\left(\frac{d}{\beta}\right)}\right)^{\frac{1}{r+1}}
            \leq \frac12. \label{eq:atten-next-rnd}
    \end{align}
    Now we have
    \begin{align*}
	\left(\frac{16\alpha^2}{d u_1 \log^{1/2} \left(\frac{d}{\beta}\right)}\right)^{\frac{1}{r+1}}
	={} &\left( \frac{1}{d} \right)^{\frac{1}{r+1}} \cdot \left(\frac{16\alpha^2}{u_1 \log^{1/2} \left(\frac{d}{\beta}\right)}\right)^{\frac{1}{r+1}} \\
	\leq{} &\frac{1}{2} \cdot \left(\frac{16\alpha^2}{u_1 \log^{1/2} \left(\frac{d}{\beta}\right)}\right)^{\frac{1}{r+1}} \tag{$r \leq R = \log_2(d/2)$} \\
	\leq{} &\frac{1}{2}
	\end{align*}
	where the last inequality holds for $d$ larger than some absolute constant.
    Therefore,~\eqref{eq:atten-next-rnd} is satisfied
    for all $1 \leq r \leq R$.
\end{proof}

Claim \ref{clm:atten-next-rnd} completes the
inductive step of the proof. It establishes
that at the beginning of round $r+1$, $p_j < u_{r+1}$
for all $j \in S_{r+1}$.  Now we can take a union bound over all the failure events in each round and over each of the $R$ rounds so that the conclusions of the Lemma hold with probability $1-O(\beta)$.
This completes the proof of Lemma \ref{lem:ppde_attenuation}.
\end{proof}

\subsubsection{Analysis of the Final Round}

In this section we show that the error of the coordinates $j$ such that $q[j]$ was set in the final round $r$ is small.

\begin{lem} \label{lem:ppde_final}
    Let $r$ be the final round for which
    $u_r |S_r| \leq 1$.  If $p[j] \leq u_r$
    for every $j \in S_{r}$, and $X^{r}$
    contains at least
    $$
    m =  \frac{128 d \log^3 (dR/\beta )}{\alpha^2} +
        \frac{128 d \log^{5/4}(d/\alpha\beta(2\rho)^{1/2})}{\alpha \rho^{1/2}}
    $$
    i.i.d.~samples from $P$, then with
    probability at least $1-O(\beta)$,
    then $\SD(P[S_r], Q[S_r]) \leq O(\alpha)$
\end{lem}

\begin{proof}
    Again, we use the notation,
    $\tilde{p}_{r} = \frac{1}{m} \sum_{i = 1}^{m} X^{r}_{i}$,
    for the rest of this proof. First we have,
    two claims that bound the difference
    between $p[j]$ and $\tilde{p}[j]$.
    \begin{clm} \label{clm:final-samp-err-lrg}
        For each $j \in S_r$, such that $p_j > \frac{1}{d}$,
        with probability at least $1-2\beta/d$,
        we have,
        $$
        \forall j \in S_{r}~~\left| p[j] - \tilde{p}_r[j] \right|
            \leq \sqrt{\frac{4p[j] \log\left( \frac{d}{\beta} \right)}{m}}
        $$
    \end{clm}
    \begin{proof}[Proof of Claim \ref{clm:final-samp-err-lrg}]
        The proof is identical to that of
        Claim \ref{clm:atten-samp-err}.
    \end{proof}
    
    \begin{clm} \label{clm:final-samp-err-sml}
        For each $j \in S_r$, such that $p_j \leq \frac{1}{d}$,
        with probability at least $1-4\beta/d$,
        we have,
        $$
        \forall j \in S_{r}~~\left| p[j] - \tilde{p}_r[j] \right|
        	\leq \frac{\alpha}{4d\log\left( \frac{d}{\beta} \right)}
        $$
    \end{clm}
    \begin{proof}[Proof of Claim \ref{clm:final-samp-err-sml}]
        We use a Chernoff Bound (Theorem \ref{thm:chernoff}), and
        facts that 
        $$\forall \gamma > 0~~~\KL(p+\gamma||p) \geq \frac{\gamma^2}
        {2(p+\gamma)}~~~\textrm{and}~~~\KL(p-\gamma||p)
        \geq \frac{\gamma^2}{2p}.$$
        There are two cases to analyze.
        \begin{itemize}
            \item $p_j > \frac{\alpha^2}
                {16d\ln^2\left(\frac{dR}{\beta}\right)}$:
                In this case, setting
                $\gamma = \sqrt{\frac{4p[j]
                \log\left( \frac{d}{\beta} \right)}{m}}$,
                we get $\gamma \leq p[j]$.
                Then we apply Theorem \ref{thm:chernoff}
                with appropriate parameters.
            \item $p_j \leq \frac{\alpha^2}
                {16d\ln^2\left(\frac{dR}{\beta}\right)}$:
                In this case, setting
                $\gamma = \frac{4\log\left( \frac{d}{\beta} \right)}{m}$,
                we get $\gamma \geq p[j]$.
                Then we apply Theorem \ref{thm:chernoff}
                with the corresponding parameters.
        \end{itemize}
        Since  $p[j] \leq \frac{1}{d}$,
        if $m$ satisfies the assumption of the lemma, then $$\max\left\{
            \sqrt{\frac{4p[j]
            \log\left( \frac{d}{\beta} \right)}{m}},
            \frac{4\log\left( \frac{d}{\beta} \right)}{m}
            \right\} \leq
            \frac{\alpha}{4d\log\left( \frac{d}{\beta} \right)}.$$
        Therefore, with high probability, the maximum
        error is
        $\frac{\alpha}{4d\log\left( \frac{d}{\beta} \right)}$.
    \end{proof}

    \begin{clm} \label{clm:final-no-trunc}
        With probability at least $
        1-\beta$,  for every $X^{r}_{i}
        \in X^{r}$, $\|X^{r}_{i}\|_{2} \leq B_{r}$,
        so no rows of $X^{r}$ are truncated in the
        computation of $\tmean_{B_{r}}(X^{r})$.
    \end{clm}
    \begin{proof}[Proof of Claim \ref{clm:final-no-trunc}]
    	We assume that all marginals specified
        by $S_r$ are upper bounded by $u_r$. Now,
        the expected value of $\|X^r_i\|_2^2$ is
        upper bounded by $1$. Also,
    	$B_r = \sqrt{6\log(m/\beta)}$.
        With this, we use a Chernoff Bound (Fact \ref{fact:chernoff})
        and get the required
        result.
    \end{proof}
    
    \begin{clm} \label{clm:final-priv-err}
    With probability at least $1-2\beta$,
    $$
    \forall j \in S_{r}~~\left| \tilde{p}[j] - q_{r}[j] \right| \leq
    	\sqrt{\frac{6
    	\log\left(\frac{m}{\beta}\right)
        \log\left(\frac{2d}{\beta}\right)}{\rho m^2}}
    $$
    \end{clm} 
    \begin{proof}[Proof of Claim \ref{clm:final-priv-err}]
    	We assume that all marginals specified
        by $S_r$ are upper bounded by $u_r$. From
        Claim \ref{clm:final-no-trunc}, we know that
        with high probability, there is no truncation,
        so, the Gaussian noise
        is added to $\tilde{p}_j$ for each $j \in S_r$.
        Therefore, the only source of error here is the
        Gaussian noise. 
        Using the standard tail bound
        for zero-mean Gaussians (Lemma~\ref{fact:gaussian-error}),
        with the following parameters,
        $$
        \sigma = \sqrt{\frac{3
	    \log\left(\frac{m}{\beta}\right)}{\rho m^2}}
        ~~~\textrm{and}~~~ t = \sqrt{2\log\left(\frac{2d}{\beta}\right)},$$
        and taking the union bound over all
        columns of the dataset in that round and
        the event of truncation, we obtain the claim.
    \end{proof}

    By the above claim, the magnitude of
    Gaussian noise added to each coordinate
    in the final round is less than,
    $$\frac{\alpha}{2d\log^{1/4}\left(d/\alpha\beta\sqrt{2\rho}\right)}.$$
    
    We partition the set $S_{r}$ into
    $S_{r,L} = \set{j \in S_{r} : p[j] \leq \frac{1}{d}}$
    and $S_{r,H} = \set{j \in S_{r} : p[j] > \frac{1}{d}}$.

    \begin{clm} \label{clm:final-high}
        $\SD(P[S_{r,H}], Q[S_{r,H}]) \leq \alpha$
    \end{clm}
    \begin{proof}[Proof of Claim~\ref{clm:final-high}]
        For every coordinate $j \in S_{r,H}$, due Claim
        \ref{clm:final-samp-err-lrg}, and from the upper
        bound on the Gaussian noise added, we know that,
        $$
        \abs{p[j] - q[j]} \leq \frac{\alpha}{\log^{1/4}
            \left(\frac{d}{\beta}\right)}\sqrt{\frac{p[j]}{d}}
            = e_{r,H}. 
        $$
        Now, we know that $p[j] > \frac{1}{d}$, and
        $e_{r,H} \leq \frac{p[j]}{2}$, when $d$ is greater
        than some absolute constant.
        So, we can bound the $\chi^2$-divergence between
        such $p[j]$ and $q[j]$ by,
        \begin{align*}
            \frac{4(p[j] - q[j])^2}{q[j]} \leq&~
                \frac{4e_{r,H}^2}{p[j] - e_{r,H}}\\
                \leq&~ \frac{8\alpha^2}{d\log^{1/2}\left(\frac{d}{\beta}\right)}\\
                \leq&~ \frac{\alpha^2}{d}.
        \end{align*}
        Thus, we have $\CS(P[S_{r,H}] \| Q[S_{r,H}])
        \leq \alpha^2$, which implies
        $\SD(P[S_{r,H}], Q[S_{r,H}]) \leq \alpha$.
    \end{proof}
    
    \begin{clm} \label{clm:final-low}
        $\SD(P[S_{r,L}], Q[S_{r,L}]) \leq \alpha$
    \end{clm}
    \begin{proof}[Proof of Claim~\ref{clm:final-low}]
        By Claim \ref{clm:final-samp-err-sml}, and from
        the upper bound on the Gaussian noise added,
        for every coordinate $j \in S_{r,L}$, we have,
        \begin{align*}
            \left| p[j] - q[j] \right| \leq&~
                2\cdot\max\left\{
                \frac{\alpha}{4d\log\left(\frac{d}{\beta}\right)},
                \frac{\alpha}{2d\log^{1/4}
                \left(d/\alpha\beta\sqrt{2\rho}\right)}\right\}\\
                \leq&~ \frac{\alpha}{d\log^{1/4}\left(d/\beta\right)}.
        \end{align*}
        Thus, by Lemma~\ref{lem:sd-ub}, we have
        $\SD(P[S_{r,L}], Q[S_{r,L}]) \leq \alpha$.
    \end{proof}

    Now, combining Claims~\ref{clm:final-high}
    and~\ref{clm:final-low}, and applying
    Lemma~\ref{lem:sd-ub} completes the proof.
    \end{proof}

\subsubsection{Putting it Together}
In this section we combine Lemmas~\ref{lem:ppde_attenuation} and~\ref{lem:ppde_final} to prove Theorem~\ref{thm:ppde_acc}.  First, by Lemma~\ref{lem:ppde_attenuation}, with probability at least $1-O(\beta)$, if $S_{A}$ is the set of coordinates $j$ such that $q[j]$ was set in any of the partitioning rounds, then 
\begin{enumerate}
\item $\SD(P[S_{A}], Q[S_{A}]) \leq O(\alpha)$ and
\item if $j \not\in S_{A}$ and $r$ is the final round, then $p[j] \leq u_{r}$.
\end{enumerate}
Next, by the second condition, we can apply Lemma~\ref{lem:ppde_final} to obtain that if $S_{F}$ consists of all coordinates set in the final round, then with probability at least $1-O(\beta)$, $\SD(P[S_{F}], Q[S_{F}]) \leq O(\alpha)$.  Finally, we use a union bound and Lemma~\ref{lem:sd-ub} to conclude that, with probability at least $1-O(\beta)$,
$$
\SD(P,Q) \leq \SD(P[S_{A}], Q[S_{A}]) + \SD(P[S_{F}], Q[S_{F}]) = O(\alpha).
$$
This completes the proof of Theorem~\ref{thm:ppde_acc}.

\section{Lower Bounds for Private Distribution Estimation}
\label{sec:lb}
In this section we prove a number of lower bounds for private distribution estimation, matching our upper bounds up to polylogarithmic factors.
For estimating the mean of product or Gaussian distributions, we prove lower bounds for the weaker notion of $(\varepsilon,\delta)$-differential privacy, but still show that they nearly match our upper bounds which are under the stronger notion of $\frac{\varepsilon^2}{2}$-zCDP.
For estimating the covariance of Gaussian distributions, our lower bound is for $\eps$-DP, a stronger notion than our upper bound, which is $\frac{\eps^2}{2}$-zCDP.
Proving lower bounds for covariance estimation with stronger privacy (i.e., concentrated or approximate differential privacy) is an interesting open question.

Our proofs generally consist of two parts.
First, we prove a lower bound on the sample complexity required for private parameter estimation.
For our lower bounds on mean estimation, we use modifications of the ``fingerprinting'' method.
Then, we show that if two distributions are distance in parameter distance (either $\ell_2$-distance between their means, or Frobenius distance between their covariances), then they will be far in statistical distance.
Though we consider questions of the latter sort to be very natural, we were surprised to find they have not been studied as sigificantly as we expected.
For example, while a lower bound on the statistical distance between Gaussian distributions in terms of the $\ell_2$-distance between their means is folklore, a bound in terms of the Frobenius distance between their covariance matrices is fairly recent~\cite{DevroyeMR18a}.
Furthermore, to the best of our knowledge, our lower bound on the statistical distance between binary product distributions in terms of the $\ell_2$-distance between their means is entirely novel.

In Section~\ref{sec:lb-product}, we describe our lower bounds for learning product distributions.
In Section~\ref{sec:lb-gaussian-mean}, we describe our lower bounds for learning Gaussian distributions with known covariance.
Finally, in Section~\ref{sec:lb-gaussian-cov}, we describe our lower bounds for learning the covariance of Gaussian distributions.

\subsection{Privately Learning Product Distributions}
\label{sec:lb-product}

In this section we prove that our algorithm for learning product distributions has optimal sample complexity up to polylogarithmic factors.  
Our proof actually shows that our algorithm is optimal even if we only require the learner to work for \emph{somewhat balanced} product distributions (i.e.~those whose marginals are bounded away from $0$ and $1$) and allow the learner to satisfy the weaker variant of $(\eps,\delta)$-DP. 
The lower bound has two steps: (1) a proof that estimating the mean of a somewhat balanced product distribution up to $\alpha$ in $\ell_2$ distance (Lemma~\ref{lem:product-fp}) and (2) a proof that estimating a somewhat balanced product distribution in total-variation distance implies estimating its mean in $\ell_2$ distance (Lemma~\ref{lem:product-param}).
Putting these two lemmata together immediately implies the following theorem:
\begin{theorem}
\label{thm:product-lb}
For any $\alpha \leq 1$ smaller than some absolute constant, any $\left(\eps, \frac{3}{64n}\right)$-DP mechanism which estimates a product distribution to accuracy $\leq \alpha$ in total variation distance with probability $\geq 2/3$ requires $n = \Omega\left(\frac{d}{\alpha\eps\log d}\right)$ samples.
\end{theorem}
\begin{proof}
We will show that no algorithm can estimate the mean of a product distribution up to accuracy $\alpha$ with probability $2/3$ with fewer than $O\left(\frac{d}{\alpha\eps\log d}\right)$ samples (for an appropriate choice of the constant in the big-Oh notation).
By Lemma~\ref{lem:product-param}, this would imply an algorithm with the same sample complexity which estimates the distribution in total variation distance up to accuracy $C\alpha$.
The theorem statement follows after a rescaling of $\alpha$.

Suppose that such an algorithm existed.
By repeating the algorithm $O(\log d)$ times, the success probability could be boosted by a standard argument\footnote{Specifically, repeat the algorithm $O(\log d)$ times, and choose any output which is close to at least half the outputs. This is correct with high probability by using the Chernoff bound and the fact that the original algorithm was accurate with probability $\geq 2/3$.} to $\geq 1 - 1/d^2$, with the overall algorithm requiring $O\left(\frac{d}{\alpha\eps}\right)$ samples.
Since the domain is bounded, any answer will be, at worst, an $O(\sqrt{d})$-accurate estimate in $\ell_2$-distance.
This implies that the expected accuracy of the resulting algorithm is at most $O(\alpha)$, which is precluded by Lemma~\ref{lem:product-fp}, for an appropriate choice of constant in the big-Oh notation.
\end{proof}

\begin{lem}
\label{lem:product-fp}
    If $M : \pmo^{n \times d} \to [-\frac13,\frac13]^d$ is
    $\left(\eps, \frac{3}{64n}\right)$-DP, and for every
    product distribution $P$ over $\pmo^{d}$
    such that $-\frac13 \preceq \ex{}{P} \preceq \frac13$,
    $$
    \ex{X \sim P^{\otimes n}}{\| M(X) -
        \ex{}{P} \|_2^2 } \leq \alpha^2 \leq \frac{d}{54}
    $$
    then $n \geq \frac{d}{72 \alpha \eps}$.
    Equivalently, if $M$ is $\left(\eps,\frac{3}{64n}\right)$-DP and
    is such that for every product distribution $P$ over
    $\zo^{d}$ such that $\frac13 \preceq \ex{}{P} \preceq \frac23$,
    $$
    \ex{X \sim P^{\otimes n}}{\| M(X) -
        \ex{}{P} \|_2^2 } \leq \alpha^2 \leq \frac{d}{216}
    $$
    then $n \geq \frac{d}{144\alpha \eps}$.
\end{lem}

\begin{proof}
We will only prove the first part of the theorem for estimation over $\pmo^{d}$, and the second part will follow immediately by a change of variables.

    Let $P^1,\dots,P^d \sim [-\frac13,\frac13]$ be
    chosen uniformly and independently from
    $[-\frac13,\frac13]$. Let $P = \Ber(P^1)
    \otimes \dots \otimes \Ber(P^d)$ be the
    product distribution with mean
    $\overline{P} = (P^1,\dots,P^d)$.  Let $X_1,\dots,X_n \sim P$
    be independent samples from this product
    distribution. Define:
    \begin{align*}
        Z_i^j &=~ \left(\frac{\frac{1}{9}-(P^j)^2}{1-(P^j)^2}\right)(M^j(X)-P^j)(X_i^j-P^j)\\
        Z_i &=~ \sum_{j=1}^{d}Z_i^j
    \end{align*}
    where $Z_i$ is a measure of the correlation between the
    estimate $M(X)$ and the $i$-th sample $X_i$.
    We will use the following key lemma, which is an
    extension of a similar statement in \cite{SteinkeU15}
    for the uniform distribution over $[-1,1]$.
    \begin{lemma}[Fingerprinting Lemma]\label{lem:fingerprinting}
        For every $f : \pmo^{n} \to [-\frac13,\frac13]$,
        we have
        $$
        \ex{P \sim [-\frac13,\frac13], X_{1 \dots n} \sim P}
            {\left(\frac{\frac{1}{9}-P^2}{1-P^2}\right) \cdot
            (f(X)-P) \cdot \sum_{i=1}^{n} (X_i - P)
            + (f(X) - P)^2} \geq \frac{1}{27}
        $$
    \end{lemma}
    \begin{proof}[Proof of Lemma \ref{lem:fingerprinting}]
        Define the function $$g(p) = \ex{X_{1\dots n}\sim p}{f(X)}.$$
        For brevity, we will write
        $\ex{P}{\cdot}$ to indicate that the
        expectation is being taken over $P$,
        where $P$ is chosen uniformly
        from $[-\frac13,\frac13]$.
        By \cite[Lemma A.1]{BunSU17}, for every fixed $p$,
\begin{equation}
        h(p) := \ex{X_{1 \dots n} \sim p}{\left(\frac{\frac{1}{9}-p^2}{1-p^2}\right)\cdot
            (f(X)-p) \cdot \sum_{i=1}^{n}(X_i-p)}
            = \left(\frac{1}{9}-p^2\right)g'(p). \label{eq:fingerprinting}
\end{equation}
        Where we have defined the function $h(p)$ for brevity.  Now we have,
        \begin{align}
            \ex{P}{h(P)} = \ex{P}{\left(\frac{1}{9}-P^2\right)g'(P)} &=~
                \frac{3}{2}\int\limits_{-1/3}^{1/3}{\left(\frac{1}{9}-p^2\right)g'(p)dp}\nonumber\\
                &=~ 2\cdot \ex{P}{P g(P)}\label{eq:fngrprnt-ex}.
        \end{align}
        Now, using the above identity, we have:
        \begin{align*}
            \ex{P,X_{1\dots n} \sim P}
                {\left(f(X)-P\right)^2}
            &=~ \ex{P,X_{1\dots n} \sim P}{f(X)^2}
                + \ex{P}{P^2} -
                2 \cdot \ex{P}{Pg(P)}\\
            &\geq~ \ex{P}{P^2} -
                2\cdot \ex{P}{Pg(P)}\\
	        &=~ \ex{P}{P^2} -
    	        \ex{P}{h(P)} \tag{Using \eqref{eq:fngrprnt-ex}}
	    \end{align*}
	    Rearranging the above inequality gives:
        $$
            \ex{P,X_{1\dots n} \sim P}{\left(f(X)-P\right)^2} + \ex{}{h(P)}
            \geq \ex{P}{P^2} = \frac{1}{27}.
        $$
    \end{proof}

	Henceforth, all expectations are taken
    over $\overline{P}$, $X$, and $M$.
    We can now apply the lemma to the function
    $M^j(X)$ for every $j \in [d]$, use
    linearity of expectation, and the accuracy
    assumption to get the bound,
    \begin{align*}
        \sum_{i=1}^{n} \ex{}{Z_i}
        &=~ \sum_{j=1}^{d}\ex{}{\sum_{i=1}^{n}Z_i^j}\\
        &\geq~ \frac{d}{27} - \ex{}{\|M(X)-\overline{P}\|_2^2}\\
        &\geq~ \frac{d}{27} - \alpha^2\\
        &\geq~ \frac{d}{54},
    \end{align*}
    where the second inequality follows from the
    assumption $\ex{}{\|M(X) - \overline{P}\|_2^2} \leq \alpha^2 \leq \frac{d}{54}$.
    
    To complete the proof, we will give an
    upper bound on $\sum_{i=1}^{n} \ex{}{Z_i}$
    that contradicts the lower bound unless
    $n$ is sufficiently large. Consider any
    $i \in [n]$. Define:
    \begin{align*}
        {\wt{Z}_i^j} &=~ \left(\frac{\frac{1}{9}-(P^j)^2}{1-(P^j)^2}\right)(M(X_{\sim i})-P^j)(X_i^j-P^j)\\
        \wt{Z}_i &=~ \sum_{j=1}^{d}{\wt{Z}_i^j}
    \end{align*}
    where $X_{\sim i}$ denotes $X$ with the
    $i$-th sample replaced with an independent
    draw from $P$. Since $X_{\sim i}$ and
    $X_i$ are conditionally independent
    conditioned on $P$, $\ex{}{\wt{Z}_i} = 0$.
    Also, we have:
    $$
        \ex{}{|\wt{Z}_i|}^2 \leq \ex{}{\wt{Z}_i^2} = \var{}{\wt{Z}_i} \leq
            \frac{1}{9}\ex{}{\| M(X) - \overline{P} \|_2^2} \leq
            \frac{1}{9}\alpha^2
    $$
    where the first inequality is Jensen's.
    Furthermore, we have the following upper bound on the maximum value of $\wt{Z}_i$ and $Z_i$: $\|Z_i \|_\infty \leq 8d/81$ and $\| \wt{Z}_i \|_{\infty} \leq 8d/81$.
    
    Now we can apply differential privacy to
    bound $\ex{}{Z_i}$, using the fact that
    $X$ and $X_{\sim i}$ differ on at most
    one sample. 
    The approach is akin to Lemma 8 of~\cite{SteinkeU17b}.
    The main idea is to split $Z_i$ into its positive and negative components $Z_{i,+}$ and $Z_{i,-}$, write each of them as $\ex{}{Z_{i,*}} = \int_{0}^{\| Z_{i} \|_\infty} \pr{}{Z_{i,*} \geq t}\,dt$, and apply the definition of $(\varepsilon,\delta)$-approximate differential privacy to relate them to the similar quantities for $\wt{Z}_i$.
    Implementing this strategy gives the following:
    \begin{align*}
        \ex{}{Z_i}
        &\leq~ \ex{}{\wt{Z}_i} + 2\eps \cdot \ex{}{|\wt{Z}_i|} +
            2\delta \cdot \| Z_i \|_{\infty}\\
        &\leq~ 0 + 2\eps \cdot \frac{1}{3}\alpha + \frac{3}{32 n}
            \cdot \frac{8d}{81}\\
        &\leq~ \frac{2}{3} \alpha \eps + \frac{d}{108 n}.
    \end{align*}
    Note that we used the upper bound $e^\varepsilon -1 \leq 2\varepsilon$ for $\varepsilon \leq 1$.
    Thus, we have:
    $$
    \sum_{i=1}^{n} \ex{}{Z_i} \leq \frac{2}{3} \alpha \eps n + \frac{d}{108}.
    $$
    
    Combining the upper and lower bounds gives:
    $$
    \frac{d}{54} \leq  \frac{2}{3}\alpha \eps n + \frac{d}{108}
        \Longleftrightarrow n \geq \frac{d}{72 \alpha \eps}.
    $$
    This completes the proof.
\end{proof}

\begin{lem}
\label{lem:product-param}
    Let $P$ and $Q$ be two product distributions
    with mean vectors $p$ and $q$ respectively,
    such that $p_i \in [1/3,2/3]$ for all $i \in [d]$.
    Suppose that 
    $$\|\ex{}{P} - \ex{}{Q}\|_2 \geq \alpha,$$
    for any $\alpha \leq \alpha_0$, where $0 < \alpha_0 \leq 1$ is some absolute constant.
    Then $\SD(P,Q) \geq C\alpha$, for some
    absolute constant, $C$.
\end{lem}
\begin{proof}
    Consider the set, $A = \{x \mid \log(P(x)/Q(x)) > \alpha\}$.
    If we show that $P(A) = \Omega(1)$, then
    we would have the following.
    \begin{align*}
        &\forall x \in A \quad \frac{P(x)}{Q(x)} > e^\alpha \geq 1+\alpha \\
        \implies{} &\forall x \in A \quad P(x) - Q(x) > \frac{\alpha}{1+\alpha}P(x) \geq \frac{\alpha}{2}P(x)\\
        \implies{} &P(A) - Q(A) > \frac{\alpha}{2}P(A)\\
        \implies{} &P(A) - Q(A) > \Omega(\alpha)\\
        \implies{} &\SD(P,Q) > \Omega(\alpha).
    \end{align*}
    To this end, let $x = (x_1,\dots,x_d) \in \{0,1\}^d$.
    Then,
    $$P(x) = \prod\limits_{i=1}^{d}{p_i^{x_i}(1-p_i)^{1-x_i}}
        ~~~\text{and}~~~
        Q(x) = \prod\limits_{i=1}^{d}{q_i^{x_i}(1-q_i)^{1-x_i}}.$$
    Therefore,
    \begin{align*}
        Z(x) \coloneqq \log(P(x)/Q(x)) =&~ \log(P(x)) - \log(Q(x))\\
            =&~ \sum\limits_{i=1}^{d}{x_i\log\frac{p_i}{q_i}}
                + \sum\limits_{i=1}^{d}{(1-x_i)\log\frac{1-p_i}{1-q_i}}\\
            =&~ -\sum\limits_{i=1}^{d}{x_i\log\frac{q_i}{p_i}}
                - \sum\limits_{i=1}^{d}{(1-x_i)\log\frac{1-q_i}{1-p_i}}.
    \end{align*}
    Now, we lower bound $Z(x)$ by some function
    of $x$, so that if that function takes a value
    larger than $\alpha$ with probability
    $\Omega(1)$ (measured with respect to $P$), then $Z(x) \geq \alpha$ with
    probability $\Omega(1)$.
    Noting that
    $\log(t) \leq t - 1$ for all $t > 0$, we get the following:
    \begin{align*}
        \log(P(x)/Q(x)) \geq&~
            \sum\limits_{i=1}^{d}{x_i\left(1-\frac{q_i}{p_i}\right)}
            + \sum\limits_{i=1}^{d}{(1-x_i)\left(1-\frac{1-q_i}{1-p_i}\right)}\\
            =~ &\sum\limits_{i=1}^{d}{\frac{(x_i-p_i)(p_i-q_i)}{p_i(1-p_i)}}
    \end{align*}
    Let $Y_i = \frac{(X_i-p_i)(p_i-q_i)}{p_i(1-p_i)}$
    be a transformation of the random variable $X_i \sim P_i$.
    To be precise, it will have the following PMF:
    \begin{equation*}
        Y_i =
        \begin{cases}
            \frac{p_i-q_i}{p_i} &\textrm{w.p.}~~~ p_i\\
            -\frac{p_i-q_i}{1-p_i} &\textrm{w.p.}~~~ 1-p_i.
        \end{cases}
    \end{equation*}
    $Y_i$ has the following properties: 
    $$\ex{}{Y_i} = 0,~~~
        \sigma_i^2 = \ex{}{Y_i^2} = \frac{(p_i-q_i)^2}{p_i(1-p_i)},
        ~~~\textrm{and}~~~
        \ex{}{Y_i^3} = \frac{(p_i-q_i)^3\left[p_i^2 + (1-p_i)^2\right]}
        {p_i^2(1-p_i)^2}.$$
    Let $\sigma^2 = \sum_{i \in [d]} \sigma_i^2$, and 
    $Y = \frac{1}{\sigma}\sum\limits_{i=1}^{d}{Y_i}$.
    Hence, $Z(x) \geq \sigma Y$ for all $x$.
    At this point, it suffices to show that
    $\pr{}{Y>\alpha/\sigma} \geq \Omega(1)$.
    We will do this in two parts: we show that if $Y$ was a Gaussian with the same mean and variance, then this inequality would hold, and we also show that $Y$ is well-approximated by said Gaussian.
    We start with the latter.

    We apply the Berry-Esseen theorem~\cite{Berry41,Esseen42,Shevtsova10} 
    to approximate the distribution of $Y$
    by the standard normal distribution.
    Let $\psi$ be the actual CDF of $Y$,
    and $\phi$ be the CDF of the standard
    normal distribution:
    \begin{align*}
        \abs{\psi(y) - \phi(y)} \leq&~ C_1 \sigma^{-1} \cdot
            \max\limits_{i}{\frac{\ex{}{Y_i^3}}{\ex{}{Y_i^2}}}\\
            =&~ C_1 \sigma^{-1} \cdot \max\limits_{i}
                \frac{(p_i - q_i)\left[p_i^2 + (1-p_i)^2\right]}{p_i(1-p_i)}\\
            \leq&~ \frac{5C_1}{2} \sigma^{-1} \cdot \max\limits_{i}(p_i - q_i)
    \end{align*}
    Here, $C_1=0.56$ is a universal constant.
    Now, we can assume that $p_i-q_i \leq C_2\alpha$ (for some constant $C_2$ of our choosing),
    otherwise $\SD(P,Q)>C_2\alpha$ trivially.
    Note that, by our assumption on $p_i \in [1/3,2/3]$, we have the following:
    \[
    \sigma^2 = \sum_{i \in [d]} \frac{(p_i - q_i)^2}{p_i(1-p_i)} \geq \frac{9}{2} \sum_{i \in [d]} (p_i - q_i)^2 = \frac{9}{2}\alpha^2.
    \]
    Therefore, $\sigma \geq 2 \alpha$, and we get the following:
    \begin{align*}
        \abs{\psi(y) - \phi(y)} \leq \frac{5C_1 C_2}{4}.
    \end{align*}
    We now use this to prove an $\Omega(1)$ lower bound on $\pr{}{Y> \alpha/\sigma}$.
    \begin{align*}
        1-\psi(\alpha/\sigma) >&~ 1-\phi(\alpha/\sigma) - \frac{5C_1C_2}{4}\\
            =&~ 1-\frac{1}{\sqrt{2\pi}}
                \int\limits_{0}^{\alpha/\sigma}{e^{-t^2}dt} - \frac{5C_1C_2}{4}\\
            \geq&~ 1/2-\frac{1}{\sqrt{2\pi}}\frac{\alpha}{\sigma}
                - \frac{5C_1C_2}{4}\\
            \geq&~ 1/2-\frac{1}{2\sqrt{2\pi}}
                - \frac{5C_1C_2}{4}\\
            \geq&~ 0.30 - \frac{5C_1C_2}{4}\\
    \end{align*}
    We want the above quantity to be
    a constant greater than zero. We could
    pick any ``small enough'' constant,
    so we pick $0.1$. Therefore, by
    choosing $C_2 < 0.16 / C_1$ (say $0.25$),
    we guarantee that $\pr{}{Y> \alpha/\sigma} > 0.1$.
    Hence, we have $\SD(P,Q) > 0.05 \alpha$,
    which completes the proof.
\end{proof}

\subsection{Privately Learning Gaussian Distributions with Known Covariance}
\label{sec:lb-gaussian-mean}

In this section, we will show a lower bound
on the number of samples required to estimate
the mean of a Gaussian distribution when
its covariance matrix is known.
The approach is similar to the product distribution case (Section~\ref{sec:lb-product}), but with modifications required for the different structure and unbounded data.

\begin{theorem}
\label{thm:guassian-lb}
For any $\alpha \leq 1$ smaller than some absolute constant, any $\left(\eps, \delta\right)$-DP mechanism (for $\delta \leq \tilde O\left(\frac{\sqrt{d}}{R n }\right)$)  which estimates a Gaussian distribution (with mean $\mu \in [-R,R]^d$ and known covariance $\sigma^2 \mathbb{I}$) to accuracy $\leq \alpha$ in total variation distance with probability $\geq 2/3$ requires $n = \Omega\left(\frac{d}{\alpha\eps\log (dR)}\right)$ samples.
\end{theorem}
While the expression for $\delta$ might seem complex, one can note that if $R = 1$ and for $d \geq 1$, we have $\delta = O\left(\frac{1}{n\sqrt{\log n}}\right)$, very similar to the statement of Lemma~\ref{lem:product-fp}.
Our statement is stronger and more general for settings of $d$ and $R$.

\begin{proof}
The proof is very similar to that of Theorem~\ref{thm:product-lb}, so we only sketch the differences.
To estimate the Gaussian to total variation distance $\alpha$, it is necessary to estimate the mean in $\ell_2$-distance to accuracy $\alpha \sigma$, evidenced by the following folklore fact (see, e.g.,~\cite{DiakonikolasKKLMS18}):
\begin{fact}
The total variation distance between $\mathcal{N}(\mu_1, \sigma^2 \mathbb{I})$ and $\mathcal{N}(\mu_2, \sigma^2 \mathbb{I})$ is at least $C\frac{\|\mu_1 - \mu_2\|_2}{\sigma}$, for an appropriate constant $C$, and all $\mu_1, \mu_2, \sigma$ such that $\frac{\|\mu_1 - \mu_2\|_2}{\sigma}$ is smaller than some absolute constant. 
\end{fact}

Similar to before, we can argue that the existence of such an algorithm implies the existence of an algorithm which is correct in expectation, at a multiplicative cost of $O(\log dR)$ in the sample complexity, as any estimate output by the algorithm is accurate up to $O(\sqrt{d}R)$ in $\ell_2$-distance.
Such an algorithm is precluded by Lemma~\ref{lem:guassian-fp} (noting that we must rescale $\alpha$ by a factor of $\sigma$), concluding the proof.
\end{proof}

\begin{lemma}
\label{lem:guassian-fp}
    If $M : \mathbb{R}^{n \times d} \to [-R\sigma,R\sigma]^d$ is
    $\left(\eps, \delta \right)$-DP for $\delta \leq \frac{\sqrt{d}}{48\sqrt{2} R n \sqrt{\log(48\sqrt{2}Rn/\sqrt{d})}}$, and for every
    Gaussian distribution $P$ with known covariance
    matrix, $\sigma^2 \mathbb{I}$,
    such that $-R \sigma \leq \ex{}{P} \leq R \sigma$,
    $$
    \ex{X \sim P^{\otimes n}}{\| M(X) -
        \ex{}{P} \|_2^2 } \leq \alpha^2 \leq \frac{d\sigma^2 R^2}{6},
    $$
    then $n \geq \frac{d\sigma}{24 \alpha \eps}$.
\end{lemma}
\begin{proof}
    By a scaling argument, we will focus on the case where $\sigma = 1$.
    We prove the following statement:
    If $M : \mathbb{R}^{n \times d} \to [-R,R]^d$ is $\left(\eps, \delta\right)$-DP for $\delta \leq \frac{\sqrt{d}}{48\sqrt{2} R n \sqrt{\log(48\sqrt{2}Rn/\sqrt{d})}}$, and for every Gaussian distribution $P$ with known covariance matrix $\mathbb{I}$ such that $-R \leq \ex{}{P} \leq R$, 
    $$
    \ex{X \sim P^{\otimes n}}{\| M(X) -
        \ex{}{P} \|_2^2 } \leq \alpha^2 \leq \frac{dR^2}{6},
    $$
    implies $n  \geq \frac{d}{24 \alpha \eps}$.

    Let $\mu^1, \dots, \mu^d$ be chosen independently
    and uniformly at random from the interval $[-R,+R]$. Let $P$ be the
    Gaussian distribution with mean vector
    $\overline \mu = (\mu^1, \dots, \mu^d)$, and covariance
    matrix $\mathbb{I}$. Let $X_1, \dots, X_n$
    be independent samples from this Gaussian
    distribution. As in the proof of the lower
    bound for product distributions, we define
    the following quantities.
    \begin{align*}
        Z_i^j &=~ \left(R^2-(\mu^j)^2\right)(M^j(X)-\mu^j)(X_i^j-\mu^j)\\
        Z_i &=~ \sum_{j=1}^{d}Z_i^j
    \end{align*}
    Again, our strategy would be to give
    upper and lower bounds on $\sum_{i=1}^{n}\ex{}{Z_i}$,
    which would conradict each other unless
    $n$ is larger than some specific quantity.
    To obtain the lower bound, we first prove
    a lemma similar to Lemma \ref{lem:fingerprinting}.
    \begin{lem}[Fingerprinting Lemma for Gaussians]
    \label{lem:fingerprinting-g}
        For every $f : \mathbb{R}^{n} \to [-R,R]$,
        we have
        $$
        \ex{\mu \sim [R,R], X_{1 \dots n} \sim \normal(\mu, 1 )}
            {\left(R^2-\mu^2\right) \cdot
            (f(X)-\mu) \cdot \sum_{i=1}^{n} (X_i - \mu)
            + (f(X) - \mu)^2} \geq \frac{R^2}{3}
        $$
    \end{lem}
    \begin{proof}[Proof of Lemma \ref{lem:fingerprinting-g}]
        Define the function
        $$g(\mu) = \ex{X_{1\dots n \sim \normal(\mu,1)}}{f(X)}.$$
        For brevity, we will write
        $\ex{\mu}{\cdot}$ to indicate that the
        expectation is being taken over $\mu$,
        where $\mu$ is chosen uniformly
        from $[-R,R]$. 
        We use an adaptation of~\eqref{eq:fingerprinting} to the Gaussian setting.
        From an extension of a similar statement in the full version of~\cite{DworkSSUV15}, 
        for every fixed $\mu$,
        $$h(\mu) := \ex{X_{1 \dots n} \sim \normal(\mu,1)}
            {\left(R^2-\mu^2\right)\cdot
            (f(X)-\mu) \cdot \sum_{i=1}^{n}(X_i-\mu)}
            = \left(R^2-\mu^2\right)
                \frac{\partial}{\partial \mu}g(\mu).$$
        Therefore, we get:
        $$\ex{\mu}{h(\mu)} = 2\ex{\mu}{\mu g(\mu)}.$$
        Now, using the above, we get:
        \begin{align*}
            \ex{\mu,X_{1\dots n} \sim \normal(\mu,1)}
                {\left(f(X)-\mu\right)^2}
            &=~ \ex{\mu,X_{1\dots n} \sim \normal(\mu,1)}{f(X)^2}
                + \ex{\mu}{\mu^2} -
                2 \cdot \ex{\mu}{\mu g(\mu)}\\
            &\geq~ \ex{\mu}{\mu^2} -
                2\cdot \ex{\mu}{\mu g(\mu)}\\
	        &=~ \ex{\mu}{\mu^2} -
                \ex{\mu}{h(\mu)}.
	    \end{align*}
        Rearranging the above inequality gives:
        $$
            \ex{\mu,X_{1\dots n} \sim \normal(\mu,1)}{\left(f(X)-\mu\right)^2} + \ex{\mu}{h(\mu)}
            \geq \ex{\mu}{\mu^2} = \frac{R^2}{3}.
        $$
    \end{proof}
    Henceforth, all expectations are taken
    over $\overline{\mu}$, $X$, and $M$.
    In the same way as in case of product
    distributions, we obtain the following
    bound,
    \begin{align}
        \sum_{i=1}^{n} \ex{}{Z_i} \nonumber
        &=~ \sum_{j=1}^{d}\ex{}{\sum_{i=1}^{n}Z_i^j}\nonumber\\
        &\geq~ \frac{dR^2}{3} - \ex{}{\|M(X)-\overline{\mu}\|_2^2}\nonumber\\
        &\geq~ \frac{dR^2}{3} - \alpha^2\nonumber\\
        &\geq~ \frac{dR^2}{6}\label{eq:gaussian-expectation-lb},
    \end{align}
    where the second inequality follows from the
    assumption $\ex{}{\|M(X) - \overline{P}\|_2^2} \leq \alpha^2 \leq \frac{dR^2}{6}$.
    Now, to give an upper bound, we first define:
    \begin{align*}
        {\wt{Z}_i^j} &=~ \left(R^2-(\mu^j)^2\right)(M(X_{\sim i})-\mu^j)(X_i^j-\mu^j)\\
        \wt{Z}_i &=~ \sum_{j=1}^{d}{\wt{Z}_i^j}
    \end{align*}
    where $X_{\sim i}$ denotes $X$ with the
    $i$-th sample replaced with an independent
    draw from $P$. Because $X_{\sim i}$ and
    $X_i$ are independent
    conditioned on $P$, $\ex{}{\wt{Z}_i} = 0$.
    Using similar calculations as in Lemma
    \ref{lem:product-fp}, we get the following.
    $$\ex{}{|\wt{Z_i}|}^2 \leq R^4\alpha^2$$
    Observe that, in contrast to Lemma~\ref{lem:product-fp}, we do not have a worst-case bound on the value of the statistic $Z_i$, as the support of $X_i^j$ is the real line, rather than just $\{\pm 1\}$ as before.
    Consequently, we split the computation of the expectation of $Z_{i,+}$ into the intervals $[0,T]$ and $(T, \infty)$, and only apply $(\eps,\delta)$-DP to the former.
    Again, we use the ideas of the same lemma
    about splitting $Z_i$ into $Z_{i,+}$ and
    $Z_{i,-}$ to get the following, for any $T > 0$.
\begin{equation}
\label{eq:gaussian-apply-dp}
        \ex{}{Z_i}
        \leq~ \ex{}{\wt{Z}_i} + 2\eps \cdot \ex{}{|\wt{Z}_i|} +
            2\delta \cdot T + 2\int_{T}^{\infty}{\pr{}{Z_{i,+}>t}dt}
\end{equation}
    Now,
    \begin{align*}
        Z_{i,+} &\leq~ \max\left\{\sum_{j=1}^{d}{(R^2-(\mu^j)^2)(X_i^j-\mu^j)(M^j(X)-\mu^j)}, 0\right\}\\
            &\leq~ \max\left\{2R^3\sum_{j=1}^{d}{(X_i^j-\mu^j)}, 0\right\}\\
            &=~ \max\left\{Y_{i}, 0\right\},
    \end{align*}
    where $Y_{i} \sim \normal(0, 4R^6 d)$.
    Let $W_i = \tfrac{Y_i}{2R^3 \sqrt{d}}$,
    and $S = \tfrac{T}{2R^3 \sqrt{d}}$. 
    This transformation results in $W_i$ being a standard normal random variable.
    We perform a change of variables, and repeatedly use the inequality $\mathrm{erfc}(x) \leq \exp\left(-x^2\right)$ in the following derivation:
    \begin{align*}
        \int_{T}^{\infty}{\pr{}{Z_{i,+}>t}dt}
            &=~ 2R^3\sqrt{d}\int_{ S}^{\infty}{\pr{}{W_i>s}ds}\\
            &\leq~ R^3\sqrt{d}
                \int_{ S}^{\infty}{e^{-s^2/2}ds}\\
            &=~ R^3\sqrt{d} \sqrt{\frac{\pi}{2}}
                \mathrm{erfc}\left(\frac{S}{\sqrt{2}}\right) \\
            &\leq~ R^3 \sqrt{\frac{d\pi}{2}}
                e^{-\frac{S^2}{2}}.\\
            &=~ R^3 \sqrt{\frac{d\pi}{2}}
                e^{-\frac{T^2}{8R^6d}}.\\
    \end{align*}
    We will upper-bound this by $\delta T$:
\[
      R^3 \sqrt{\frac{d\pi}{2}} e^{-\frac{T^2}{8R^6d}} \leq \delta T
\]
Equivalently, 
\[
      \frac{R^3}{\delta} \sqrt{\frac{d\pi}{2}}  \leq T e^{\frac{T^2}{8R^6d}}
\]
Consider setting $T = 2\sqrt{2}R^3\sqrt{d}\sqrt{\log(1/\delta)}$.
The right-hand side of this inequality becomes 
\[
2\sqrt{2}R^3\sqrt{d}\sqrt{\log(1/\delta)} \cdot \frac{1}{\delta},
\]
which is greater than the left-hand side for any $\delta < 1$.

Using~\eqref{eq:gaussian-apply-dp}, this gives us the following upper bound:
    $$\sum_{i=1}^{n}{\ex{}{Z_i}} \leq 2R^2 \alpha \eps n +
    4\sqrt{2}R^3\sqrt{d}n\delta \sqrt{\log(1/\delta)}  $$
On the other hand,~\eqref{eq:gaussian-expectation-lb} gives us a lower bound on this quantity, and thus we require that the following inequality is satisfied:
\[
\frac{dR^2}{12} \leq n\left( R^2\alpha \eps  +
    2\sqrt{2}R^3\sqrt{d}\delta \sqrt{\log(1/\delta)}   \right)
\]
Our goal is to find conditions on $\delta$ such that the product involving this term is at most $\frac{dR^2}{24}$.
If this holds, the corresponding term can be moved to the left-hand side, and we are left with the following inequality:
\[
\frac{dR^2}{24} \leq nR^2\alpha \eps,
\]
which is satisfied when
\[
n \geq \frac{d}{24\alpha\eps},
\]
as we desired.

Thus, it remains to find conditions on $\delta$ which satisfy the following inequality:
\[
     2\sqrt{2}R^3\sqrt{d}\delta \sqrt{\log(1/\delta)} \leq \frac{dR^2}{24n}.
\]
Rearranging, we get:
\[
     \delta \sqrt{\log(1/\delta)} \leq \frac{\sqrt{d}}{48\sqrt{2}Rn}\triangleq \phi.
\]
Consider setting $\delta = \frac{\phi}{2\sqrt{\log(1/\phi)}}$.
This results in 
\begin{align*}
     \delta \sqrt{\log(1/\delta)} &\leq \frac{\phi}{2\sqrt{\log(1/\phi)}} \cdot \sqrt{\log(1/\phi) + \log(2\sqrt{\log (1/\phi)})} \\
&= \frac{\phi}{2} \cdot \sqrt{1 + \frac{\log(2\sqrt{\log (1/\phi)})}{\log (1/\phi)}}\\
&\leq \phi,
\end{align*}
where the last inequality is because $\sqrt{1 + \frac{2\sqrt{x}}{x}} \leq 2$ for all $x \geq 0$.
\end{proof}

\subsection{Privately Learning Gaussian Distributions with Unknown Covariance}
\label{sec:lb-gaussian-cov}
In this section, we prove lower bounds for privately learning a Gaussian with unknown covariance.

\begin{theorem}
\label{thm:guassian-cov-lb}
For any $\alpha \leq 1$ smaller than some absolute constant, any $\eps$-DP mechanism which estimates a Gaussian distribution to accuracy $\leq \alpha$ in total variation distance with probability $\geq 2/3$ requires $n = \Omega\left(\frac{d^2}{\alpha\eps}\right)$ samples.
\end{theorem}
\begin{proof}
The proof is again similar to that of Theorem~\ref{thm:product-lb}, and we sketch the differences.
The primary difference is that instead of considering algorithms which estimate the covariance matrix of the distribution in Frobenius norm, we consider algorithms which estimate the \emph{inverse} of the covariance matrix.
The reason is the following theorem of~\cite{DevroyeMR18a}, which states that if one fails to estimate the inverse of the covariance matrix of a Gaussian in Frobenius norm, then one fails to estimate the Gaussian in total variation distance:
\begin{theorem}[Theorem 3.8 of~\cite{DevroyeMR18a}]
\label{thm:gauss-cov-approx}
Suppose there are two mean-zero Gaussian distributions $\mathcal{N}_1$ and $\mathcal{N}_2$, with covariance matrices $\Sigma_1$ and $\Sigma_2$, respectively.
Furthermore, suppose that $\Sigma_1^{-1} - I$ and $\Sigma_2^{-1} - I$ are both zero-diagonal and have Frobenius norm smaller than some absolute constant.
Then the total variation distance between $\mathcal{N}_1$ and $\mathcal{N}_2$ is at least $c_1 \|\Sigma_1^{-1} - \Sigma_2^{-1}\|_F - c_2 (\|\Sigma_1^{-1} - I\|_F^2 + \|\Sigma_2^{-1} - I\|_F^2)$, for constants $c_1, c_2 > 0$.
\end{theorem}
Therefore, it suffices to show that there does not exist an algorithm which estimates the inverse of the covariance in Frobenius norm with probability $\geq 2/3$, where the inverse of the covariance matrix obeys the conditions of Theorem~\ref{thm:gauss-cov-approx}.
As before, an algorithm which is accurate in expectation would imply the existence of such an algorithm, so we show that such an algorithm does not exist.
We do this by applying a modification of Lemma~\ref{lem:gaussian-cov-dp}.
While this lemma is stated in terms of estimating the covariance matrix, we can obtain an identical statement for estimating the inverse of a covariance matrix by repeating the argument, with $\Sigma$ replaced by $\Sigma^{-1}$ at all points.
Note that the construction in Lemma~\ref{lem:gaussian-cov-dp} obeys the conditions of Theorem~\ref{thm:gauss-cov-approx}.
Furthermore, the Frobenius norm diameter of the construction is $\Theta(\alpha)$ (rather than $\poly(d)$ as in Theorem~\ref{thm:product-lb}), we do not lose an $O(\log d)$ factor when converting to an algorithm which is accurate in expectation.
Therefore, the application of this modification completes the proof.
\end{proof}

\begin{lemma}
    \label{lem:gaussian-cov-dp}
    If $M : \mathbb{R}^{n \times d} \to S$ is
    $\eps$-DP (where $S$ is the space of all $d \times d$ symmetric positive semi-definite matrices), and for every $\mathcal{N}(0,\Sigma)$ over $\mathbb{R}^{d}$
    such that $\frac{1}{2}\mathbb{I} \preceq \Sigma \preceq \frac{3}{2}\mathbb{I}$,
    $$
    \ex{X \sim \mathcal{N}(0,\Sigma)^{\otimes n}}{\| M(X) -
        \Sigma \|_F^2 } \leq \alpha^2/64,
    $$
    then $n \geq \Omega\left(\frac{d^2}{ \alpha \eps}\right)$.
\end{lemma}
\begin{proof}
Let $P$ be the uniform distribution over the set of $d \times d$ symmetric matrices with $0$ on the diagonal, where the $d^2 - d$ non-zero entries are $\left\{\pm \frac{\alpha}{2d}\right\}$.
Note that there are $(d^2 - d)/2$ free parameters, and thus $2^{(d^2 - d)/2}$ matrices.
For each $v \in \mathrm{supp}(P)$, we will define $\Sigma(v) = I + v$.
It is easy to check that for all $v \in \mathrm{supp}(P)$, that $\frac{1}{2}\mathbb{I} \preceq \Sigma \preceq \frac{3}{2}\mathbb{I}$, and furthermore, that for any $v, v' \in \mathrm{supp}(P)$, $\dtv\left(\mathcal{N}(0, \Sigma(v)), \mathcal{N}(0,\Sigma(v')\right) \leq O(\alpha)$ (Lemma~\ref{lem:gaussian-tv}).
We assume that our algorithm is aware of this construction, and thus will always output a symmetric matrix with $1$'s on the diagonal and off-diagonal entries bounded in magnitude by $\alpha/2d$.

Define the random variables $Z$ and $Z'$, which are sampled according to the following process.
Let $V$ and $V'$ be independently sampled accordingly to $P$.
$X$ is a set of $n$ samples from $\mathcal{N}(0, \Sigma(V))$, and similarly, $X'$ is a set of $n$ independent samples from $\mathcal{N}(0, \Sigma(V'))$.
Then, $M(X)$ and $M(X')$ are computed with their own (independent) randomness.
We then define:
$$Z = \langle M(X), V\rangle  = 2 \sum_{i < j} M_{ij}(X) \cdot V_{ij}, $$
$$Z' = \langle M(X'), V\rangle  = 2 \sum_{i < j} M_{ij}(X') \cdot V_{ij}. $$

We start with the following claim which lower bounds the expectation of $Z$.
\begin{claim}
\label{clm:z-claim}
$\ex{}{Z} \geq \frac{\alpha^2}{16} - \frac{1}{2}\| M(X) - \Sigma(V) \|_F^2 \geq \frac{7\alpha^2}{128}$.
\end{claim}
\begin{proof}
\begin{align*}
\ex{}{2Z + \| M(X) - \Sigma(V) \|_F^2 } &= \sum_{i < j} \ex{}{4 M_{ij}(X) V_{ij}  + 2(M_{ij}(X) - \Sigma_{ij}(V))^2} \\
&= \sum_{i < j} \ex{}{2M^2_{ij}(X) + 2\Sigma^2_{ij}(V)} \\
&\geq \sum_{i < j} \ex{}{ 2\Sigma^2_{ij}(V)} \\
&= \sum_{i < j} \frac{\alpha^2}{2d^2} \\
&= \frac{\alpha^2}{2} \cdot \frac{d^2 - d}{2d^2} \\
&\geq \frac{\alpha^2}{8},
\end{align*}
where the last inequality holds for any $d \geq 2$.
The claimed statement follows by rearrangement, and the second inequality by the assumption on $\|M(X) - \Sigma(V)\|_F^2$.
\end{proof}

Next, we show that $Z'$ will not be too large, with high probability.
\begin{claim}
\label{clm:z'-claim}
$\pr{}{Z' > \alpha^2/32} \leq \exp(-\Omega(d^2))$.
\end{claim}
We begin by observing $M(X')$ and $V$ are independent.
Condition on any realization of $M(X')$. 
Then $Z'\mid M(X')$ is the sum of $(d^2 - d)/2$ independent summands, each with contained in the range $[-\alpha^2/4d^2,\alpha^2/4d^2]$ and with expectation $0$ (since $\ex{}{V_{ij}} = 0$).
By Hoeffding's inequality, we have that 
$$\pr{}{Z' > \alpha^2/32 \mid M(X')} \leq \exp\left(-\frac{2\frac{\alpha^4}{1024}}{\frac{d^2-d}{2} \cdot \frac{\alpha^4}{4d^4}}\right) \leq \exp\left(-\Omega(d^2)\right). $$
The claim follows by noting that value of $M(X')$ which we conditioned on was arbitrary.
\end{proof}

\begin{claim}
\label{clm:both-z-claim}
$\pr{}{Z > \alpha^2/32} \leq \exp(O(\alpha \eps n)) \cdot \pr{}{Z' > \alpha^2/32}.$
\end{claim}
\begin{proof}
We start by proving the following lemma:
\begin{lemma}
Let $M : \mathbb{R}^{n \times d} \rightarrow \mathcal{A}$ be an $\eps$-DP mechanism.
Suppose that $D, D'$ are probability distributions over $\mathbb{R}^d$ such that $\dtv(D, D') = \alpha$.
Then for $S \subseteq \mathcal{A}$,
$$\pr{M, X \sim D^n}{M(X) \in S} \leq \exp(O(\eps \alpha n))\pr{M, X' \sim D'^n}{M(X') \in S}.$$
\end{lemma}
\begin{proof}
By the definition of $D$ and $D'$, this implies that there exists distributions $A, B, C$ such that
\begin{align*}
D &= \alpha B + (1- \alpha)A \\
D' &= \alpha C + (1 - \alpha)A
\end{align*}
We will actually prove the following for any subset $S \subseteq \mathcal{A}$:
\[\pr{M, X \sim D^n}{M(X) \in S} \leq \exp(O(\eps \alpha n))\pr{M, X' \sim A^n}{M(X') \in S}.\]
A symmetric argument, with $D'$ in place of $D$, and using the other direction of the definition of differential privacy will give the lemma statement (with an extra factor of $2$ in the exponent).

We will draw $X,X'$ from a coupling of $(D,A)$.
In particular, let $W \in \{0,1\}^n$ be a random vector, where each entry is independently set to be $1$ with probability $\alpha$ and $0$ otherwise.
Note that $Y(w) \triangleq \sum_i w_i$ is distributed as $\mathrm{Bin}(n,\alpha)$.
Then there exists a coupling $C$ of $(D,A)$ such that $X' \sim A^n$, and $X_i$ is equal to $X'_i$ when $W_i = 1$, and is an independent draw from $B$ otherwise.
\begin{align*}
\pr{\substack{X \sim D^n \\ M}}{M(X) \in S} &= \sum_{w} \pr{}{W = w} \pr{\substack{(X,X') \sim C \\ M} }{M(X) \in S \mid w}  \\
&\leq \sum_{w} \pr{}{W = w} \exp\left(\eps Y(w)\right) \pr{\substack{(X,X') \sim C \\ M} }{M(X') \in S \mid w}  \\
&= \pr{\substack{X' \sim A^n\\ M} }{M(X') \in S} \sum_{w} \pr{}{W = w} \exp\left(\eps Y(w)\right)  \\
&= \pr{\substack{X' \sim A^n \\ M} }{M(X') \in S} \ex{}{\exp\left(\eps Y(w)\right)} \\
&= (1 - \alpha + \alpha e^\eps)^n \pr{\substack{X' \sim A^n \\ M} }{M(X') \in S} \\
&\leq \exp\left(\alpha (e^\eps - 1)n\right) \pr{\substack{X' \sim A^n\\ M} }{M(X') \in S} \\
&\leq \exp\left(O\left(\alpha \eps n \right)\right) \pr{\substack{X' \sim A^n \\ M} }{M(X') \in S},
\end{align*}
as desired.
The first inequality uses the definition of differential privacy.
We also used the moment generating function of the binomial distribution, and the fact that $e^\eps - 1 = O(\eps)$ for $\eps < 1$.
\end{proof}
With this in hand, the proof is as follows:
\begin{align*}
&\pr{\substack{V,V' \sim P \\  X \sim \mathcal{N}(0, \Sigma(V)) \\ M}}{Z > \alpha^2/32} \\
=~&\sum_{v, v' } \pr{V, V' \sim P}{V = v, V' = v'}\pr{\substack{X \sim \mathcal{N}(0, \Sigma(V)) \\ M}}{Z > \alpha^2/32 \mid V, V'} \\
\leq~&\exp(O( \alpha \eps n)) \sum_{v, v'} \pr{V, V' \sim P}{V = v, V' = v'}\pr{\substack{X \sim \mathcal{N}(0, \Sigma(V')) \\ M}}{Z' > \alpha^2/32 \mid V, V'} \\
=~&\exp(O( \alpha \eps n)) \pr{\substack{V,V' \sim P \\  X \sim \mathcal{N}(0, \Sigma(V')) \\ M}}{Z' > \alpha^2/32} 
\end{align*}
The inequality follows using the lemma above, and noting that for any $v, v' \in \mathrm{supp}(P)$, that $\dtv(\mathcal{N}(0, \Sigma(V)), \mathcal{N}(0, \Sigma(V'))) \leq O(\alpha)$.
\end{proof}

With this in hand, we make the following observations.
Claim~\ref{clm:z-claim} implies that $\Omega(1) \leq \pr{}{Z > \alpha^2/32}$.
Claim~\ref{clm:z'-claim} states that $\pr{}{Z' > \alpha^2/32} \leq \exp(-\Omega(d^2))$.
Using these together with Claim~\ref{clm:both-z-claim} gives us that $\Omega(1) \leq \exp(O(\alpha \eps n) - \Omega(d^2))$, which implies that we require $n \geq \Omega(d^2/\alpha\eps)$ to avoid a contradiction.

\addcontentsline{toc}{section}{References}
\bibliographystyle{alpha}
\bibliography{biblio}

\appendix

\ifnum\epsdeltastuff=1
\input{missingproofs}
\fi

\end{document}